\documentclass[journal,twoside,web]{ieeecolor}

\usepackage{generic}
\usepackage{graphicx}      
\usepackage{tabulary}
\usepackage{amsfonts}
\usepackage{amsmath}

\usepackage[frak=euler,scr=boondoxo,scrscaled=1.05]{mathalfa}

\usepackage{amssymb}
\usepackage{pgf,tikz}
\usepackage{tabularx}
\usepackage{float}
\usepackage{cite}
\usepackage{stfloats}
\usepackage{xcolor}
\usepackage{comment}

\let\labelindent\relax
\usepackage{enumitem} 
\usepackage{bbm} 
\usepackage{bm}

\usepackage[algo2e,ruled,vlined,linesnumbered,resetcount,norelsize]{algorithm2e}

\newtheorem{remark}{Remark}
\newtheorem{theorem}{Theorem}
\newtheorem{lemma}{Lemma}
\newtheorem{assumption}{Assumption}

\newtheorem{definition}{Definition}
\newtheorem{corollary}{Corollary}

\newcommand{\col}{{\rm col\;}}

\newcommand{\row}{{\rm row\;}}

\newcommand{\longthmtitle}[1]{\mbox{} \emph{(#1):}}

\SetCommentSty{mycommfont}

\def\qed{ \rule{.1in}{.1in}}

\def\BibTeX{{\rm B\kern-.05em{\sc i\kern-.025em b}\kern-.08em
		T\kern-.1667em\lower.7ex\hbox{E}\kern-.125emX}}
\begin{document}

\title{Resilience for Distributed Consensus with Constraints}
\author{Xuan Wang, Shaoshuai Mou, and Shreyas Sundaram \vspace{-1.5em}
	\thanks{This work is supported by the NASA University Leadership Initiative (ULI) under grant number 80NSSC20M0161 and the National Science Foundation CAREER award 1653648. X. Wang is with the Department of Electrical and Computer Engineering, George Mason University, Fairfax, VA 22030 USA, {xwang64@gmu.edu}. S. Mou is with the School of Aeronautics and Astronautics, Purdue University, West Lafayette, IN 47906 USA, {mous@purdue.edu}.  S. Sundaram is with the Elmore Family School of Electrical and Computer Engineering, Purdue University, West Lafayette, IN 47906 USA, {sundara2@purdue.edu}}}
\maketitle

\begin{abstract}                          
	This paper proposes a new approach that  enables multi-agent systems to achieve resilient \textit{constrained} consensus in the presence of Byzantine attacks, 
	in contrast to existing literature that is only applicable to \textit{unconstrained} resilient consensus problems. 
	The key enabler for our approach is a new device called a \textit{$(\gamma_i,\alpha_i)$-resilient convex combination}, which allows normal agents in the network to utilize their locally available information to automatically isolate the impact of the Byzantine agents.
	Such a resilient convex combination is computable through linear programming, whose complexity scales well with the size of the overall system. By applying this new device to multi-agent systems, we introduce network and constraint redundancy conditions under which resilient constrained consensus can be achieved with an exponential convergence rate. 
	We also provide insights on the design of a network such that the redundancy conditions are satisfied.
	Finally, numerical simulations and an example of safe multi-agent learning are provided to demonstrate the effectiveness of the proposed results.
\end{abstract}
	
\begin{IEEEkeywords}
	Resilient distributed coordination; Multi-agent consensus     
\end{IEEEkeywords}

\section{Introduction}
The emergence of \textit{multi-agent} systems is motivated by modern applications that require more efficient data collection and processing capabilities in order to achieve a higher-level of autonomy\cite{FJS09DC}. For multi-agent systems, the
main challenge lies in the requirement of a scalable control algorithm that allows all agents in the system to interact in a cooperative manner. Towards this end, traditional approaches equipped with a centralized coordinator can suffer from poor performance due to limitations on the coordinator's computation and communication capabilities. 
Motivated by this, significant research effort has been devoted to developing \textit{distributed} control architectures\cite{WR08book,GN17TCNS,ZSL16Auto}, which enable multi-agent systems to achieve global objectives through only local coordination. 
A specific example is \textit{consensus}, where each agent controls a local state which it repeatedly updates as a weighted average of its neighbors' states \cite{MAB08TAC,MAB08SIAM,XMT16Auto,QSY17TAC}. In this way, the states of all agents will gradually approach the same value \cite{PSJW19ARC,TXJ19ARC,XPY16TAC}, allowing the system to achieve certain global control or optimization goals. Recently, the consensus approach has powered many applications, including motion synchronization \cite{MM10DC}, multi-robot path planning/formation control \cite{XMT15TCNS}, flocking of mobile robots \cite{CH19CSL},  and cooperative sensing \cite{GBU17TAC,XJSM19TAC} across the overall network.

While consensus algorithms are simple and powerful, their effectiveness heavily depends on the assumption that all agents in the network correctly execute the algorithm. 
In practice, however, the presence of large numbers of heterogeneous agents in multi-agent systems (provided by different vendors, with differing levels of trust) introduces the potential for some of the agents to be malicious. These agents may inject misinformation that propagates throughout the network and prevent the entire system from achieving proper consensus-based coordination. As shown in \cite{Shreyas18TAC}, even one misbehaving agent in the network can easily lead the states of other agents to any value it desires.  Traditional information-security approaches \cite{MH11PIS} (built on the pillars of confidentiality, integrity, and availability) focus on protecting the data itself, but are not directly applicable to consensus-based distributed algorithms for cooperative decision making, as such processes are much more sophisticated and vulnerable than simple data transmission. Because of this, instead of narrowly focusing on techniques such as data encryption and identity verification \cite{MH11PIS}, researchers are motivated to seek fundamentally new approaches that allow the network to reliably perform consensus even after the adversary has successfully compromised certain agents in the system. Towards this end, the main challenges arise from the special characteristics of distributed systems where no agent has access to global information. As a consequence, classical paradigms such as fault-tolerant control \cite{MC97PIEE} are usually not sufficient to capture the sophisticated actions (i.e., Byzantine failures\footnote{Byzantine failures are considered the most general and most difficult class of failures among the failure modes. They imply no restrictions, which means that the failed agent can generate arbitrary data, including data that makes it appear like a functioning (normal) agent.}) the adversarial agents take to avoid detection.

\textit{Literature review:} Motivated by the need to address security for consensus-based distributed algorithms, recent research efforts have aimed to establish more advanced approaches that are inherently resilient to Byzantine attacks, i.e., allowing the normal agents to utilize only their local information to automatically mitigate or isolate misinformation from attackers.
Along this direction, relevant literature dates back to the results in \cite{M83ACM,G84ACM}, which reveal that a resilient consensus can be achieved among $m$ agents with at most $(m-1)/3$ Byzantine attackers. 
However, these techniques require significant computational effort by the agents (and global knowledge of the topology).
To address this, the Mean Subsequence Reduced (MSR) algorithms were developed \cite{Shreyas13Selected,SH17Auto}, based on the idea of letting each normal agent run distributed consensus by excluding the most extreme values from its neighbor sets at each iteration. Under certain conditions on the network topology, the effectiveness of these algorithms can be theoretically validated. Followed by \cite{Shreyas13Selected,SH17Auto}, similar approaches have been further applied to resilient distributed optimization \cite{KLS20Arxiv,JWM20SAS} and multi-agent machine learning\cite{PRJ17NIPS}. 
Note that in the results of \cite{Shreyas13Selected,SH17Auto}, the local states for all agents must be scalars. If the agents control multi-dimensional states, one possible solution is to run the scalar consensus separately on each entry of the state. However, such a scheme will violate the convex combination property of distributed consensus (i.e., the agents' states in the new time-step must be located within the convex hull of its neighbors' states at the current time-step.) and thus, lead to a loss of the spatial information encoded therein. 
In order to maintain such convex combination property among agents' states, methods based on Tverberg points have been proposed for both centralized\cite{mendes2015DC} and distributed cases\cite{N14ICDCN,PH15IROS,PH16ICRA,PH17TOR}. However, according to \cite{WD13DCG,PME08ACM}, the computational complexity of calculating exact Tverberg points is typically high. To reduce the computational complexity, new approaches based on \textit{resilient convex combinations}\cite{LN14ACM,HMNV15DC,XSS18NACO,JYXC20IFAC} have been developed, which can be computed via linear or quadratic programming with lower complexity. 

It is worth pointing out that the results mentioned above are applicable to \textit{unconstrained} resilient consensus problems. To the best of the authors' knowledge, there is no existing literature that expends these results to constrained consensus problems where the consensus value must remain within a certain set.
Such problems play a significant role in the field of distributed control and optimization\cite{AAP10TAC,SJA15TAC,XSB19TAC}. This motivates the goal of this paper: developing a fully distributed algorithm that can achieve resilient constrained consensus for multi-agent systems with exponential convergence rate. The key challenges in such settings are as follows. (\textbf{a}) To satisfy local constraints, each agent must update its local state by a resilient convex combination of its neighbors' states, precluding the use of entry-wise consensus. (\textbf{b}) The proposed approach should guarantee an exponential convergence rate. 
(\textbf{c}) For constrained consensus problems, to isolate the impact of malicious information from Byzantine agents, in addition to the network redundancy, one also needs constraint redundancy.  

\textit{Statement of Contributions}: We study the resilient consensus algorithm for multi-agent systems with local constraints. Originated from our previous preliminary work in \cite{XSS18NACO}, we introduce a new device named \textit{$(\gamma_i,\alpha_i)$-Resilient Convex Combination}. Given a maximum number of Byzantine neighbors for each agent, we introduce an approach that allows the multi-agent system to achieve resilient constrained consensus with an exponential convergence rate.  We provide conditions on network and constraint redundancies under which the effectiveness of the proposed algorithm can be guaranteed. 
We observe that the proposed network redundancy condition is difficult to verify directly. Thus, we introduce a sufficient condition, which is directly verifiable and offers an easy way to design network typologies satisfying the required condition. 
All the results in the paper are validated with examples in un-constrained consensus, constrained consensus, and safe multi-agent learning.	
Our result differs from existing literature in two major aspects. Firstly, our approach has an exponential convergence rate, which is not guaranteed in our previous work \cite{XSS18NACO}. The related work in \cite{N14ICDCN} first introduced a special way of using Tverberg points for resilient consensus and achieved an exponential convergence rate. 
Despite the similar properties shared by our approach and the one in \cite{N14ICDCN}, the key mechanism of our approach is different from using Tverberg points. Furthermore, as discussed in Fig. 1 of \cite{XSS18NACO}, when a Tverberg point does not exist, a \textit{Resilient Convex Combination} always exists.
Secondly, compared with \cite{XSS18NACO,N14ICDCN,LN14ACM,HMNV15DC,JYXC20IFAC,PH15IROS,PH16ICRA,PH17TOR} that are designed only for unconstrained consensus, our new approach is capable of achieving resilient constrained consensus. The involvement of constraints requires the introduction of extra redundancy conditions to guarantee convergence.
It is worth mentioning that the redundancy conditions we derived to handle constraints are not restricted to the \textit{$(\gamma_i,\alpha_i)$-Resilient Convex Combination} proposed in this paper, but can also be integrated into the results in \cite{XSS18NACO,N14ICDCN,LN14ACM,HMNV15DC,JYXC20IFAC,PH15IROS,PH16ICRA,PH17TOR} by applying the same redundancy conditions accordingly.

The remaining parts of the paper are organized as follows. Section \ref{Sec_Prob} presents the mathematical formulation of the distributed resilient constrained-consensus problem. In Section  \ref{SecNRCC}, we introduce a new device named $(\gamma_i,\alpha_i)$-Resilient Convex Combination, and propose an algorithm that allows the normal agents in the network to obtain such a resilient convex combination only based on its local information. Section \ref{Sec_CTC} presents how the $(\gamma_i,\alpha_i)$-Resilient Convex Combination can be applied to solve resilient constrained consensus problems, and specifies under which conditions (network and constraint redundancies) the effectiveness of the proposed algorithm can be theoretically guaranteed.
Numerical simulations and an example of safe multi-agent learning are given in Section \ref{Sec_SM} to validate the results. We conclude the paper in Section \ref{Sec_CL}.

\textit{Notations:} Let $\mathbb{Z}$ denote the set of integers and $\mathbb{R}$ denote the set of real numbers. Let $x_{\mathcal{A}}=\{x_j\in \mathbb{R}^n~,~j\in \mathcal{A}\}$ denote a set of states in $\mathbb{R}^n$, where $\mathcal{A}$ is a set of agents. 
Given $x_1, x_2, \cdots, x_r \in \mathbb{R}^{n}$, let $\col\{x_j|~j=1,2,\cdots,r\}=\begin{bmatrix}
	x_1^{\top}&x_2^{\top}&\cdots &x_r^{\top}
\end{bmatrix}^{\top}\in\mathbb{R}^{rn}$ denote a column stack of vectors;
let $\row\{x_j|~j=1,2,\cdots,r\}=\begin{bmatrix}
	x_1&x_2&\cdots &x_r
\end{bmatrix}\in\mathbb{R}^{n\times r}$ denote a row stacked matrix of vectors. Let $\mathscr{V}(\mathbb{G})$ denote the agent (vertex) set of a graph $\mathbb{G}$.
Let ${\bf 1}_r$ denote a vector in $\mathbb{R}^r$ with all its components equal to 1. Let $I_r$ denote the $r\times r$ identity matrix. 
For a vector $\beta$, by $\beta> 0$ and $\beta\geq 0$, we mean that each entry of vector $\beta$ is positive and non-negative, respectively.

\section{Problem formulation}\label{Sec_Prob}
Consider a network of ${m}$ agents in which each agent $i$ is able to receive information from certain other agents, called agent $i$'s neighbors. Let $\mathcal{N}_i(t)$, $i\in\{1,\cdots,m\}$ denote the set of agent $i$'s neighbors at time $t$. The neighbor relations can be characterized by a time-dependent graph $\mathbb{G}(t)$ such that there is a directed edge from $j$ to $i$ in $\mathbb{G}(t)$ if and only if $j\in \mathcal{N}_i(t)$. For all $i$, $t$, we assume $i\in{\mathcal{N}}_i(t)$. Let $m_i(t)=|{\mathcal{N}}_i(t)|$, which is the cardinality of the neighbor set.
Based on network $\mathbb{G}(t)$, in a consensus-based distributed algorithm, each agent is associated with a local state $x_i\in\mathbb{R}^n$; and the goal is to reach a consensus for all agents, such that
\begin{align}\label{Csed_Csus}
	x_1=&\cdots=x_m=x^*,
\end{align}
where $x^*$ is a common value of interest.

\subsection{Consensus-based Distributed Algorithms}\label{subsec_CBDA}
In order to achieve the $x^*$ in \eqref{Csed_Csus}, consensus-based distributed algorithms usually require each agent to update its state by an equation of the form
\begin{align}
	x_i(t+1)=f_i(v_i(t),x_i(t)),  \label{eq_mainalgorithm}
\end{align}
where $f_i(\cdot)$ is an operator associated with the problem being tackled by the agents. For example, it may arise from constraints (as a projection operator), objective functions  (as a gradient descent operator), and so on. The quantity $v_i(t)$ in \eqref{eq_mainalgorithm} is a convex combination of agent $i$'s neighbors' states, i.e.,
\begin{align}\label{eq_oconvexcomb}
	v_i(t)=\sum\limits_{j\in\mathcal{N}_i(t)} w_{ij}(t)x_j(t).
\end{align}
Here, $w_{ij}(t)\ge0$ is a weight that agent $i$ places on its neighboring agent $j\in\mathcal{N}_i(t)$, with $\sum_{j\in\mathcal{N}_i(t)} w_{ij}(t)=1$ for all $i\in\{1,\cdots,m\}$ and $t=1,2,3,\cdots$. Note that different choices of $w_{ij}(t)$ can lead to different $v_i(t)$, which can influence the convergence of the algorithm. 
One feasible choice of $w_{ij}(t)$ can be obtained by letting $ w_{ij}(t)=\frac{1}{|\mathcal{N}_i(t)|}$, which uniformly shares the weights among the neighbors \cite{SJA15TAC}.

With the $v_i(t)$ in \eqref{eq_oconvexcomb}, we further specify the form of update \eqref{eq_mainalgorithm} for solving constrained consensus problems. Suppose the agents of the network are subject to local convex constraints $x_i\in\mathcal{X}_i$, $i\in\{1,\cdots,m\}$, with $\bigcap_{i=1}^m\mathcal{X}_i\neq \emptyset$. To obtain a consensus value in \eqref{Csed_Csus} such that 
\begin{align} %
	x^*&\in\bigcap_{i=1}^m \mathcal{X}_i,\label{Csed_Csus2}
\end{align}
update \eqref{eq_mainalgorithm} can take the form
\begin{align}
	x_i(t+1)
	=\mathcal{P}_i(v_i(t)),  \label{eq_algCC}
\end{align}
where $\mathcal{P}_i(\cdot)$ is a projection operator that projects any state to the local constraint $\mathcal{X}_i$ \cite{SJA15TAC}. 
Note that if the local constraint does not exist, i.e. $\mathcal{X}_i=\mathbb{R}^n,~ \forall i\in\{1,\cdots,m\}$, then $\mathcal{P}_i(v_i(t))=v_i(t)$ is an identity-mapping.

\subsection{Distributed Consensus Under Attack}\label{subsec_Byzantine}
Suppose the network $\mathbb{G}(t)$ is attacked by a number of Byzantine agents, which are able to connect themselves to the normal agents in $\mathbb{G}(t)$ in a time-varying manner, and can send arbitrary state information to those connected agent(s). Note that if one Byzantine agent is simultaneously connected to multiple normal agents, it does not necessarily have to send the same state information to all of those agents.
We use ${\mathbb{G}}^+(t)$ to characterize the new graph where both normal and Byzantine agents are involved. Accordingly, for normal agents $i=1,\cdots,m$, we use ${\mathcal{M}}_i(t)$ to denote the set of its Byzantine neighbors.
For all $t=1,2,3,\cdots$,  we assume the number of each agent's Byzantine neighbors is  $|{\mathcal{M}}_i(t)|= {\kappa}_i(t)$. Further for all $t=1,2,\cdots$ and $i\in\{1,\cdots,m\}$, we assume ${\kappa}_i(t)\le\bar{\kappa}\in\mathbb{Z}_+$ is bounded. Note that the normal agents do not know which, if any, of their neighbors are Byzantine.

Under network ${\mathbb{G}}^+(t)$, if one directly applies the constrained consensus algorithm, the malicious states sent by Byzantine agents will be injected into the convex combination \eqref{eq_oconvexcomb} as
\begin{align}\label{eq_mconvexcomb}
	{v}_i(t)
	=&\sum\limits_{j\in{\mathcal{N}}_i(t)} {w}_{ij}(t)x_j(t)+\sum\limits_{k\in{\mathcal{M}}_i(t)} {w}_{ik}(t)x_{ik}(t),
\end{align}
where $x_{ik}(t)$ is the state information sent by Byzantine agent $k$ to the normal agent $i$,
${w}_{ij}(t),{w}_{ik}(t)\ge0$, $i\in\{1,\cdots,m\}$, $j\in\mathcal{N}_i(t)$, $k\in{\mathcal{M}}_i(t)$, and 
$\sum_{j\in{\mathcal{N}}_i(t)} {w}_{ij}(t)+\sum_{k\in{\mathcal{M}}_i(t)} {w}_{ik}(t)=1.$
According to equation \eqref{eq_mconvexcomb}, since $x_{ik}(t)$ can take arbitrary values, as long as  ${w}_{ik}(t)$ are not all equal to zero, the Byzantine agents can fully manipulate  ${v}_i(t)$ to any desired value and therefore prevent the normal agents from reaching a consensus.

\subsection{The Problem: Resilient Constrained Consensus}
In this paper, \textit{the problem of interest} is to develop a novel algorithm such that each normal agent in the system can mitigate the impact of its (unknown) Byzantine neighbors when calculating its local consensus combination; we call this a resilient convex combination.
Then by employing such a resilient convex combination into update \eqref{eq_algCC}, all agents in the system are able to achieve resilient \textit{constrained} consensus with \textit{exponential} convergence rate even in the presence of Byzantine attacks.

\section{An Algorithm to Achieve the $(\gamma_i,\alpha_i)$-Resilient Convex Combination} \label{SecNRCC}

Since the convex combination ${v}_i(t)$ in \eqref{eq_mconvexcomb} is the key reason why Byzantine agents can compromise a consensus-based distributed algorithm,  it motivates us to propose a distributed approach for the normal agents in the network to achieve the following objectives. \textbf{(i)} The approach should allow the normal agents to choose a ${v}_i^\star(t)$ subject to \eqref{eq_mconvexcomb}, such that ${w}_{ik}^\star(t)=0$, for $\forall k\in\mathcal{M}_i(t)$. This guarantees that all information from Byzantine agents are automatically removed from ${v}_i^\star(t)$. 
\textbf{(ii)} The obtained ${v}_i^\star(t)$ should guarantee that, for normal neighbors of agent $i$, $j\in{\mathcal{N}}_i(t)$, a certain number of ${w}_{ij}^\star(t)$ is lower bounded by a positive constant.
This guarantees that a consensus among normal agents can be reached with a proper convergence rate.


\subsection{Resilient Convex Combination} \label{SecNRCC_KI}
In line with the two objectives proposed above, for each normal agent $i$, we define the following concept, named a \textit{$ (\gamma_i,\alpha_i)$-resilient convex combination}:

\begin{definition}\label{def_Res} \textnormal{\textbf{($(\gamma_i,\alpha_i)$-Resilient Convex Combination)}}
	A vector ${v}_i^\star(t)$ computed by agent $i\in\{1,\cdots,m\}$ is called a \textit{resilient convex combination} of its neighbor's states if it is a convex combination of its \textit{normal} neighbors' states, i.e., there exists ${w}_{ij}^\star(t)\in[0,1]$, $j\in{{\mathcal{N}}}_i(t)$, such that
	\begin{align}\label{eq_reconvexcomb}
		{v}_i^\star(t)=&\sum\limits_{j\in{{\mathcal{N}}}_i(t)} {w}_{ij}^\star(t)x_j(t)
	\end{align}
	with $\sum_{j\in{\mathcal{N}}_i(t)} {w}_{ij}^\star(t)=1.$
	Furthermore, for $\gamma_i\in\mathbb{N}$ and $\alpha_i\in(0,1)$, the resilient convex combination is called a $(\gamma_i,\alpha_i)$-\textit{resilient}, if at least $\gamma_i$ of ${w}_{ij}^\star(t)$ satisfy ${w}_{ij}^\star(t)\ge\alpha_i$.
\end{definition}

Moving forward, let ${\mathcal{N}}^+_i(t)\triangleq\mathcal{N}_i(t)\bigcup \mathcal{M}_i(t)$ be the neighbor set of $i$, which is composed of both normal and malicious agents.

\begin{assumption} \label{ass1}
	Suppose for $\forall t=1,2,3,\cdots$, and $\forall i\in\{1,\cdots,m\}$, the neighbor sets satisfy $|{\mathcal{N}}^+_i(t)|\ge (n+1){\kappa}_i(t) + 2$, where $n$ is the dimension of the state and $\kappa_i(t)$ is the number of the Byzantine neighbors of agent $i$. 
\end{assumption}

Note that this assumption will be used later in the theorems to guarantee the existence of the resilient convex combination. In addition, if  $\kappa_i(t)$ is unknown, one can replace the $\kappa_i(t)$ by an upper bound of $\kappa_i(t)$ for all $\forall i,t$, denoted by $\bar{\kappa}$.
To continue, define the following concept,

\begin{definition} [Convex hull]
	Given a set of agents $\mathcal{A}$, and the corresponding state set $x_{\mathcal{A}}$, the convex hull of vectors in $x_{\mathcal{A}}$ is defined as
	\begin{align}\label{def_CH}
		\mathcal{H}(x_{\mathcal{A}})=\left\{\sum_{j=1}^{|x_\mathcal{A}|}\omega_{j}x_j:x_j\in x_\mathcal{A}, \omega_{j}\geq 0, \sum_{j=1}^{|x_\mathcal{A}|}\omega_{j}=1\right\}.
	\end{align}
\end{definition}

\subsection{An Algorithm for Resilient Convex Combination} \label{SecProof}
The above definition allows us to present Algorithm \ref{Algorithm_Main} for agent $i$ to achieve a $(\gamma_i,\alpha_i)$-resilient convex combination, which is only based on the agent's locally available information. Note that since all the results in this section only consider the states for a particular time step $t$,  for the convenience of presentation, we omit the time-step notation $(t)$ in the following algorithm description.
\begin{algorithm2e}
	\label{Algorithm_Main}
	\caption{Compute a $(\gamma_i,\alpha_i)$-Resilient Convex Combination ${v}_i^\star$ for agent $i$.}
	\SetAlgoLined
	\textbf{Input}  ${\mathcal{N}}^+_i$, $x_{{\mathcal{N}}^+_i}$, and ${\kappa}_i$.\\
	\textit{Let} ${d}_i=|{\mathcal{N}}^+_i|$.\\
	\textit{Choose} an integer  $\sigma_i\in\left[n{\kappa}_i + 2\ ,\  {d}_i-{\kappa}_i\right]$   \tcp*{Note that ${d}_i\ge (n+1){\kappa}_i + 2$, so the set is non-empty.}
	\textit{Let} $s_i=\binom {{d}_i-1} {\sigma_i-1}$. {Construct} all possible subsets $\mathcal{S}_{ij}\subset {\mathcal{N}}^+_i$, $j\in\{1,\cdots,s_i\}$, such that $i\in\mathcal{S}_{ij}$, and $|\mathcal{S}_{ij}|=\sigma_i$ \tcp*{An illustration of sets $\mathcal{S}_{ij}$ is given in Fig. \ref{Fig_Alg_1}.} \label{st_con}
	
	\For{$j= 1,\cdots,s_i$\label{st_start}}
	{
		\textit{Define} $\widehat{\mathcal{S}}_{ij}\triangleq{\mathcal{S}}_{ij}\setminus \{i\}$.\\
		\eIf{$x_i\in\mathcal{H}(x_{\widehat{\mathcal{S}}_{ij}})$}
		{\textit{Let} \quad $\phi_{ij}=x_{i}$.\label{st_a}}
		{\textit{Let} $b_i=\binom {\sigma_i-1} {\sigma_i-{\kappa}_i-1}=\binom {\sigma_i-1} {{\kappa}_i}$. For $k\in\{1, \cdots, b_i\}$, construct sets $\mathcal{B}_{ijk}\subset \mathcal{S}_{ij}$ such that  $i \in \mathcal{B}_{ijk}$,  and $|\mathcal{B}_{ijk}|=\sigma_i-{\kappa}_i$.\\
			\textit{Let} \quad $\displaystyle \phi_{ij}\in\left(\bigcap_{k=1}^{b_i}\mathcal{H}(x_{\mathcal{B}_{ijk}})\right)\bigcap
			\mathcal{H}(x_{\widehat{\mathcal{S}}_{ij}})$ \label{st_b}
			\tcp*{ Later, we will show that the set defined here is non-empty.}} 
	}\label{st_end}
	\textit{Compute} \quad $\displaystyle {v}_i^\star =\frac{\sum_{j=1}^{s_i}\phi_{ij}+x_i}{s_i+1}$.
\end{algorithm2e}
\begin{theorem}\longthmtitle{Validity of Algorithm \ref{Algorithm_Main}}\label{Th_1}
	Suppose Assumption \ref{ass1} holds. For agent $i$, given an integer $\sigma_i\in\left[n{\kappa}_i + 2\ ,\  {d}_i-{\kappa}_i\right]$, by  Algorithm \ref{Algorithm_Main}, the obtained ${v}_i^\star$ is a $(\gamma_i,\alpha_i)$-resilient convex combination, where $\gamma_i=d_i-\kappa_i-\sigma_i+2$ and $\alpha_i= \frac{1}{(s_i+1)(n+1)}$. Here, $d_i=|{\mathcal{N}}^+_i|$ and $s_i=\binom {{d}_i-1} {\sigma_i-1}$.
\end{theorem}

\begin{remark}\label{R_Th1_2}
	Note that for a certain time step $t$, if agent $i$ has $d_i$ neighbors, then by running Algorithm \ref{Algorithm_Main}, the obtained  ${v}_i^\star$ is a convex combination of at least $\gamma_i=d_i-\kappa_i-\sigma_i+2$ normal neighbors. 
	This means in order to completely isolate the information from $\kappa_i$ malicious agents, the agent $i$ may simultaneously lose the information from $\sigma_i-2$ of its normal neighbors. \hfill$\square$
\end{remark}

\smallskip

To better understand Algorithm \ref{Algorithm_Main} before proving Theorem \ref{Th_1}, we provide Fig. \ref{Fig_Alg_1} to visualize its execution process. In a 2 dimensional ($n=2$) space, consider an agent $i=1$, which has $d_i=6$ neighbors, denoted by agents $1-6$. Suppose $\kappa_i=1$, and let $\sigma_i=5\in\left[n{\kappa}_i + 2\ ,\  {d}_i-{\kappa}_i\right]$. Then for step \ref{st_con} of Algorithm \ref{Algorithm_Main}, one can construct $s_i=\binom{5}{4}=5$ sets $\mathcal{S}_{ij}$. For example, $\mathcal{S}_{11}=\{1,2,3,4,5\}$, $\mathcal{S}_{12}=\{1,2,3,4,6\}$, and so forth. Consider $j=1$ in step \ref{st_start}. Then $\widehat{\mathcal{S}}_{11}={\mathcal{S}}_{11}\setminus\{i\}=\{2,3,4,5\}$, and accordingly, the set $\mathcal{H}(x_{\widehat{\mathcal{S}}_{11}})$ must be the yellow area in the top middle figure of Fig. \ref{Fig_Alg_1}. Since in the given example, one has $x_i\notin \mathcal{H}(x_{\widehat{\mathcal{S}}_{11}})$, then by running step 10, we have $b_1=\binom {\sigma_i-1} {{\kappa}_i}=\binom {4} {1}=4$, and
there holds $|\mathcal{B}_{11k}|=\sigma_i-\kappa_i=4$, for $k\in\{1,2,3,4\}$. Specially, $\mathcal{B}_{111}=\{1,2,3,4\}$,  $\mathcal{B}_{112}=\{1,3,4,5\}$, $\mathcal{B}_{113}=\{1,2,4,5\}$, and $\mathcal{B}_{114}=\{1,2,3,5\}$. Consequently, $\bigcap_{k=1}^{b_i}\mathcal{H}(x_{\mathcal{B}_{ijk}})$ will be the green area in the top middle figure of Fig. \ref{Fig_Alg_1}. The intersection $\left(\bigcap_{k=1}^{b_i}\mathcal{H}(x_{\mathcal{B}_{ijk}})\right)\bigcap
\mathcal{H}(x_{\widehat{\mathcal{S}}_{ij}})$ is non-empty. One can pick any point in this intersection to be $\phi_{ij}$. With the obtained $\phi_{ij}$, the output of Algorithm \ref{Algorithm_Main} can be computed as ${v}_i^\star=\frac{\sum_{j=1}^{s_i}\phi_{ij}+x_i(t)}{s_i+1}$.
\noindent\begin{figure}[t]
	\vspace{-0.1cm}
	\centering
	\includegraphics[width=8 cm]{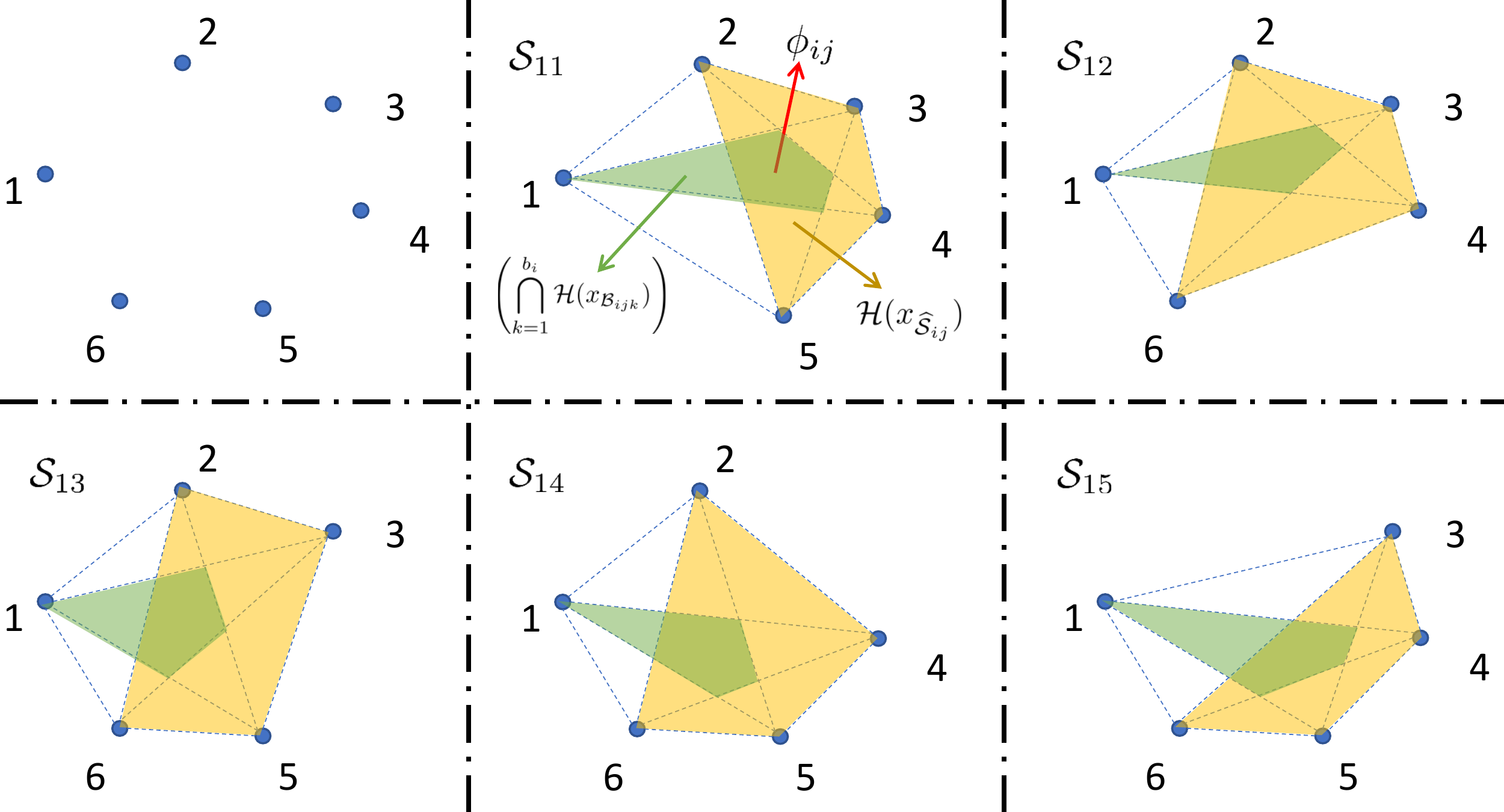}
	\caption{Visualization of Algorithm \ref{Algorithm_Main}.}
	\label{Fig_Alg_1}
\end{figure}

As a further remark, each iteration of Algorithm \ref{Algorithm_Main} (from step \ref{st_start} to step \ref{st_end}) aims to find a vector  $\phi_{ij}\in\left(\bigcap_{k=1}^{b_i}\mathcal{H}(x_{\mathcal{B}_{ijk}})\right)\bigcap
\mathcal{H}(x_{\widehat{\mathcal{S}}_{ij}})$. As will be shown in the following subsection, the set $\left(\bigcap_{k=1}^{b_i}\mathcal{H}(x_{\mathcal{B}_{ijk}})\right)$ guarantees that the chosen $\phi_{ij}$ is a resilient convex combination defined in definition \ref{def_Res}, and the set $\mathcal{H}(x_{\widehat{\mathcal{S}}_{ij}})$  guarantees certain bounds on the values of $\gamma_i$ and $\alpha_i$.

{\subsection{Proof of Theorem \ref{Th_1}} \label{SecAnalysis}}
Ahead of proving Theorem \ref{Th_1}, we first present a lemma to summarize some useful results with proofs given in the Appendix.

\begin{lemma}\label{LM_ERW}
	Algorithm \ref{Algorithm_Main} has the following properties:
	\begin{enumerate}[label=\textbf{\alph*}.]	
	\item $\textit{[Existence]}$ For the step 11 of Algorithm \ref{Algorithm_Main}, 
	\begin{align}
		\left(\bigcap_{k=1}^{b_i}\mathcal{H}(x_{\mathcal{B}_{ijk}})\right)\bigcap
		\mathcal{H}(x_{\widehat{\mathcal{S}}_{ij}})\neq\emptyset.
	\end{align} 
	
	\item $\textit{[Resilient combination]}$  For all $j\in\{1,2,\cdots,s_i\}$, the vector $\phi_{ij}$ is a resilient convex combination such that 
	\begin{align}\label{RCCphi}
		{\phi}_{ij}=\sum_{\ell\in{\widetilde{\mathcal{S}}_{ij}}}{\xi}_{ij\ell}x_\ell,\quad \sum_{\ell\in{\widetilde{\mathcal{S}}_{ij}}}{\xi}_{ij\ell}=1,\quad \xi_{ij\ell}\ge 0,
	\end{align}
	where $\widetilde{\mathcal{S}}_{ij}\triangleq {\mathcal{S}}_{ij}\setminus \mathcal{M}_i$. Note that $\widetilde{\mathcal{S}}_{ij}$ precludes any Byzantine agents from the neighbor set ${\mathcal{S}}_{ij}$. Particularly, if ${\mathcal{S}}_{ij}\bigcap \mathcal{M}_i=\emptyset$, then $\widetilde{\mathcal{S}}_{ij}={\mathcal{S}}_{ij}$.

	\item $\textit{[Weighted combination]}$  For all $ j\in\{1,2,\cdots,s_i\}$, the vector ${\phi}_{ij}$ can be written as 
	\begin{align}\label{eq_weight}
		{\phi}_{ij}=\sum_{\ell\in{\widehat{\mathcal{S}}_{ij}}}{\beta}_{ij\ell}x_\ell,\quad \sum_{\ell\in{\widehat{\mathcal{S}}_{ij}}}{\beta}_{ij\ell}=1,\quad {\beta}_{ij\ell}\ge 0,
	\end{align}
	such that there exists at least one $\ell^*\in\widehat{\mathcal{S}}_{ij}$ with ${\beta}_{ij\ell^*}\ge\frac{1}{n+1}.$ Here, as defined in Algorithm \ref{Algorithm_Main},  $\widehat{\mathcal{S}}_{ij}\triangleq{\mathcal{S}}_{ij}\setminus \{i\}$.
	\end{enumerate}	
\end{lemma}

\smallskip
\begin{remark}\label{Rm_2}
	Note that Lemma \ref{LM_ERW}\textbf{a} establishes the existence of the vector $\phi_{ij}$. 
	Then \textbf{b} and \textbf{c} represent this same vector using different combinations of $x_{\ell}$ to characterize the different properties of ${\phi}_{ij}$. In equation \eqref{RCCphi}, the set $\widetilde{\mathcal{S}}_{ij}$ includes agent $i$ and precludes malicious agents. This set aims to obtain a resilient convex combination. 
	In contrast, in equation \eqref{eq_weight}, the set ${\widehat{\mathcal{S}}_{ij}}$ precludes agent $i$ but may potentially involve Byzantine agents. This set aims to obtain a convex combination with certain weights that are bounded away from zero for some $\ell\neq i$. 
	Specifically, if ${\widehat{\mathcal{S}}_{ij}}\subset\widetilde{\mathcal{S}}_{ij}$, then the $\xi_{ij\ell}$ in \eqref{RCCphi} can correspondingly be  replaced by the $\beta_{ij\ell}$ in \eqref{eq_weight}, and leaving the remaining weights for ${\widetilde{\mathcal{S}}_{ij} \setminus \widehat{\mathcal{S}}_{ij}}$ to be 0.
	It is worth mentioning that the first intersection of convex hulls in Lemma \ref{LM_ERW}\textbf{a} can be considered as generalizations of the Lemma 3 in \cite{HMNV15DC}. More specifically, \cite{HMNV15DC} does not use each agent's own state (i.e., $x_i$) to construct a resilient convex combination. In contrast, in this paper, we relax the feasible region of $\bigcap_{k=1}^{b_i}\mathcal{H}(x_{\mathcal{B}_{ijk}})$ by taking advantage of the fact that normal agents' can always trust their own states. \hfill$\square$
\end{remark}

\noindent{\bf Proof of Theorem \ref{Th_1}:}\\
Recall  ${v}_i^\star=\frac{1}{s_i+1}\left(\sum_{j=1}^{s_i}\phi_{ij}+x_i\right)$ and $\forall j\in\{1,2,\cdots,s_i\}$, $\widetilde{\mathcal{S}}_{ij}\subset \mathcal{N}_i$, where $\mathcal{N}_i$ is the set of normal neighbors of agent $i$. Then based on Lemma  \ref{LM_ERW}\textbf{b}, one has
\begin{align}\label{the_eq2}
	{v}_i^\star=\sum_{\ell\in{\mathcal{N}}_i}w_{i\ell}x_\ell,\quad \sum_{\ell\in{\mathcal{N}}_i}w_{i\ell}=1, \quad w_{i\ell}\ge 0.
\end{align}
Equation \eqref{the_eq2} means ${v}_i^\star$ is a resilient convex combination, and particularly, for the weights $w_{i\ell}$, one has $w_{ii}\ge\frac{1}{s_i+1}>\frac{1}{(s_i+1)(n+1)}$. For $l\neq i$, from \eqref{RCCphi}, one has
\begin{align}\label{wilbijl}
	w_{i\ell}=\frac{1}{s_i+1}\sum_{ j\in\{1,2,\cdots,s_i\}}{\xi}_{ij\ell},
\end{align} 
where for $\ell\in\widetilde{\mathcal{S}}_{ij}$, ${\xi}_{ij\ell}$ is the weight defined in \eqref{RCCphi}, and for $\ell\notin\widetilde{\mathcal{S}}_{ij}$, ${\xi}_{ij\ell}=0$. To complete the proof of Theorem \ref{Th_1}, we need to further show that for $\ell\neq i$, 
at least $\gamma_i=d_i-\kappa_i-\sigma_i+1$ of weights $w_{i\ell}$ are no less than $\alpha_i= \frac{1}{(s_i+1)(n+1)}$.  
Recall that $\xi_{ij\ell}\ge 0$, and thus $\forall j\in\{1,\cdots,s_i\}, \forall\ell\in{\mathcal{N}}_i, $, 
one has 
\begin{align}\label{ieq_wbeta}
	w_{i\ell}\ge\frac{1}{s_i+1}\xi_{ij\ell}.
\end{align}
The key idea is to find the lower bound of $w_{i\ell}$ by determining the lower bound of the corresponding $\xi_{ij\ell}$. 
To proceed, define a set $\Gamma_i(0)\triangleq\left\{\mathcal{S}_{ij}~\big|~\mathcal{S}_{ij}\subset\mathcal{N}_i~ \text{and} ~j\in\{1,\cdots,s_i\}\right\}$. Note that in Algorithm \ref{Algorithm_Main}, $\mathcal{S}_{ij}$ is originally defined as subsets of $\mathcal{N}_i^+=\mathcal{N}_i\bigcup\mathcal{M}_i$. Thus, here, $\Gamma_i(0)$ is composed of the sets $\mathcal{S}_{ij}$ by precluding any $\mathcal{S}_{ij}$ that includes at least one Byzantine agent. We have $|\Gamma_i(0)|=\binom{d_i-\kappa_i-1}{\sigma_i-1}$. Now consider the following process:

\noindent\textbf{$\clubsuit$ For $k=0,1,2,\cdots$, if the set $\Gamma_i(k)$ is non-empty:}

{\setlength{\leftskip}{0.3cm}
	\noindent Pick an arbitrary $\mathcal{S}_{ij^{\star}}$ from $\Gamma_i(k)$. From Lemma \ref{LM_ERW} \textbf{c}, the corresponding ${\phi}_{ij^{\star}}$ can be written as 
	\begin{align}\label{eq_weightst}
		{\phi}_{ij}=\sum_{\ell\in{\widehat{\mathcal{S}}_{ij^{\star}}}}{\beta}_{ij^{\star}\ell}x_\ell,\quad \sum_{\ell\in{\widehat{\mathcal{S}}_{ij^{\star}}}}{\beta}_{ij^{\star}\ell}=1,\quad {\beta}_{ij^{\star}\ell}\ge 0,
	\end{align}
	such that there exists at least one $\ell^*\in\widehat{\mathcal{S}}_{ij^*}$ with ${\beta}_{ij^{\star}\ell^*}\ge\frac{1}{n+1}.$ Here, since $i\notin\widehat{\mathcal{S}}_{ij^*}$, one has $\ell^*\neq i$.
	Further note that $\mathcal{S}_{ij^{\star}}\subset \mathcal{N}_i$, then $\widetilde{\mathcal{S}}_{ij^{\star}}= {\mathcal{S}}_{ij^{\star}}\setminus \mathcal{M}_i={\mathcal{S}}_{ij^{\star}}$. 
	Thus, $\widehat{\mathcal{S}}_{ij^{\star}}\subset \widetilde{\mathcal{S}}_{ij^{\star}}$ and according to Remark \ref{Rm_2}, one can replace the ${\xi}_{ij^{\star}\ell}$ in equation \eqref{RCCphi} by the  ${\beta}_{ij^{\star}\ell}$ in \eqref{eq_weightst}. Bringing this into equation \eqref{ieq_wbeta} leads to, for $\ell^*\neq i$,
	\begin{align}\label{ieq_wbeta2}
		w_{i\ell^*}\ge\frac{\xi_{ij^{\star}\ell^*}}{s_i+1}=\frac{{\beta}_{ij^{\star}\ell^*}}{s_i+1}\ge\frac{1}{(s_i+1)(n+1)}.
	\end{align}
	Now, update $\Gamma_i(k+1)=\Gamma_i(k)\setminus\left\{\mathcal{S}_{ij}~\big|~\ell^*\in\mathcal{S}_{ij}\right\}$, which removes any $\mathcal{S}_{ij}$ composed of agent $\ell^*$. Here, $|\Gamma_i(k)|=\binom{d_i-\kappa_i-1-k}{\sigma_i-1}$.
	Note that by doing this, the new $\ell^*$ to be obtained in the new iteration will be different from the ones obtained in previous iterations. Update $k$ with $k+1$.
	
	\setlength{\leftskip}{0pt}}
\noindent\textbf{End, and iterate the process from $\clubsuit$}

Now, consider  $\widehat{k}=(d_i-\kappa_i-\sigma_i)$, one has $|\Gamma_i(\widehat{k})|=\binom{d_i-\kappa_i-1-(d_i-\kappa_i-\sigma_i)}{\sigma_i-1}=1$. Thus, there is only one $\mathcal{S}_{ij}$ in set $\Gamma_i(\widehat{k})$. Clearly, $\Gamma_i(\widehat{k}+1)=\emptyset$. This indicates that the process described in $\clubsuit$ can be repeated for $d_i-\kappa_i-\sigma_i+1$  times, from $k=0$ to $k=d_i-\kappa_i-\sigma_i$.
Consequently, in equation \eqref{the_eq2}, for $\ell\neq i$, there exist at least $\gamma_i=d_i-\kappa_i-\sigma_i+1$ number of weights $w_{i\ell}\ge\frac{1}{(s_i+1)(n+1)}$. This together with the fact that $w_{ii}\ge\frac{1}{s_i+1}$ completes the proof of Theorem \ref{Th_1}.  \hfill\qed

\begin{remark} \label{R_sigma}
	In Theorem \ref{Th_1}, the values of $\gamma_i$ and $\alpha_i$ can be changed by tuning $\sigma_i$. 
	The choice of $\sigma_i$ also influences the computational complexity for calculating ${v}_i^\star$. We will provide more discussion on this in the next subsection. \hfill$\square$
\end{remark}

\subsection{Computational Complexity Analysis } \label{Compvsgamma}
We analyze the computational complexity denoted by $C_P$ of Algorithm \ref{Algorithm_Main}, which is mainly determined by steps \ref{st_start} to \ref{st_end}. Among these steps, the major computation comes from the selection of vector $\phi_{ij}$ from the set $\left(\bigcap_{k=1}^{b_i}\mathcal{H}(x_{\mathcal{B}_{ijk}})\right)\bigcap
\mathcal{H}(x_{\widehat{\mathcal{S}}_{ij}})$ in step \ref{st_b}. To do this, one can directly compute the vertex representation of the set by intersecting convex hulls in the $n$-dimensional space. However, given a large $b_i$, the convex combination may have numerous vertices. To avoid computing all of them, here, we show that a feasible $\phi_{ij}$ can be obtained by solving a linear programming problem (LP), as shown in Algorithm \ref{Algorithm_compute_v}.

\begin{algorithm2e}[h]
	\label{Algorithm_compute_v}
	\caption{Compute $\phi_{ij}$ by Linear Programming.}
	\SetAlgoLined
	\textbf{Input} $x_{\mathcal{B}_{ijk}}$, $b_i$ and $x_{\widehat{\mathcal{S}}_{ij}}$.\\
	Define variables $\bm{\mu}_{ij}=\col\{\mu_{ijk}|~ k=1,\cdots,b_i\}\in\mathbb{R}^{b_i|\mathcal{B}_{ijk}|}$ with each $\mu_{ijk}\in\mathbb{R}^{|\mathcal{B}_{ijk}|}$; and $\lambda_{ij}\in\mathbb{R}^{|\widehat{\mathcal{S}}_{ij}|}$. \\

	Define matrix $\bm{X}=\row\{x_\ell|~\ell\in\widehat{\mathcal{S}}_{ij}\}\in\mathbb{R}^{n\times |\widehat{\mathcal{S}}_{ij}|}$.
	For all $ k\in\{1,\cdots,b_i\}$, define matrices $\bm{X}_k=\row\{x_\ell|~\ell\in\mathcal{B}_{ijk}\}\in\mathbb{R}^{n\times |\mathcal{B}_{ijk}|}$.

Set equality constraint $\bm{1}^{\top}\lambda_{ij}=1$. For all $k\in\{1,\cdots,b_i\}$, set equality constraints $\bm{1}^{\top}\mu_{ijk}=1$  and
	$\bm{X} \lambda_{ij}-\bm{X}_k \mu_{ijk}=0$. \label{alg2_step_eq}\\
	
	Set inequality constraints $\bm{\mu}_{ij}\ge 0$; and $\lambda_{ij}\ge0$. \label{alg2_step_neq}\\
	
	Choose arbitrary vectors $\eta_1\in\mathbb{R}^{{b_i|\mathcal{B}_{ijk}|}}$ and $\eta_2\in\mathbb{R}^{|\widehat{\mathcal{S}}_{ij}|}$ such that $||\eta_1||_2=||\eta_2||_2=1$. Define a linear objective function $g(\bm{\mu}_{ij},\lambda_{ij})=\eta_1^{\top}\bm{\mu}_{ij}+\eta_2^{\top}\lambda_{ij}$.\\
	
	Run linear programming with objective function $g(\bm{\mu}_{ij},\lambda_{ij})$ subject to the defined equality/inequality constraints.\\
	
	Compute $\phi_{ij}=\sum_{\ell\in{\widehat{\mathcal{S}}_{ij}}}[\lambda_{ij}]_{\ell}x_{\ell}$.\\
	
\end{algorithm2e}
\begin{remark}
	In Algorithm \ref{Algorithm_compute_v}, we propose a linear program subject to the constraints of two convex polytopes,  namely $\bigcap_{k=1}^{b_i}\mathcal{H}(x_{\mathcal{B}_{ijk}})$ and $\mathcal{H}(x_{\widehat{\mathcal{S}}_{ij}})$. 
	Note that the equality constraints in line \ref{alg2_step_eq} and the inequality constraints in line \ref{alg2_step_neq} allow the elements in $\mu_{ijk}$, $ k\in\{1,\cdots,b_i\}$ and $\lambda_{ij}$ to specify convex combinations of $\bm{X}_k$ and $\bm{X}$, respectively. Then the intersection is characterized by $\bm{X} \lambda_{ij}-\bm{X}_k \mu_{ijk}=0$. 	
	Since we only need to obtain a feasible point from the intersection of the two polytopes, the objective function $g(\bm{\mu}_{ij},\lambda_{ij})$ can be chosen arbitrarily. \hfill$\square$
\end{remark}

Now, let $C_L$ denote the computational complexity for this linear program. According to \cite{YA15ASFCS}, one has
$$C_L=\mathcal{O}\left( (\text{nz}(W)+d_v^2) \sqrt{d_v}\right), $$
where $d_v$ is the number of variables, and $\text{nz}(W)$ is the non-zero entries of a constraint matrix $W$ that encodes both equality and inequality constraints.
From Algorithm \ref{Algorithm_compute_v}, the definitions of $|\mathcal{B}_{ijk}|$ and $b_i$, we know 
\begin{align}\label{eq_dv}
	d_v&=b_i|\mathcal{B}_{ijk}|+|\widehat{\mathcal{S}}_{ij}|=\binom {\sigma_i-1} {\kappa_i}(\sigma_i-\kappa_i)+(\sigma_i-1),
\end{align}
and
$$\text{nz}(W)=2d_v+nb_i(|\mathcal{B}_{ijk}|+|\widehat{\mathcal{S}}_{ij}|),$$
where $2d_v$ entries are associated with equality and inequality constraints on $\mu_{ijk}$, $ k\in\{1,\cdots,b_i\}$ and $\lambda_{ij}$, and $nb_i(|\mathcal{B}_{ijk}|+|\widehat{\mathcal{S}}_{ij}|)$ entries are associated with equality constraints $\bm{X} \lambda_{ij}-\bm{X}_k \mu_{ijk}=0$ for all $k\in\{1,\cdots,b_i\}$. 
Thus,
\begin{align*}
	\mathcal{C}_L&=\mathcal{O}\left( (2d_v+nb_i(|\mathcal{B}_{ijk}|+|\widehat{\mathcal{S}}_{ij}|)+d_v^2) \sqrt{d_v}\right)\\
	&=\mathcal{O}\left( (2d_v+2nd_v+d_v^2) \sqrt{d_v}\right)\\
	&=\mathcal{O}\left((n+d_v)d_v^{1.5}\right)	
\end{align*}
The second equation holds because $\sigma_i\ge n\kappa_i + 2$, thus $\mathcal{O}(b_i|\mathcal{B}_{ijk}|)=\mathcal{O}(b_i(\sigma_i-\kappa_i))=\mathcal{O}(d_v)$ and $\mathcal{O}(b_i|\widehat{\mathcal{S}}_{ij}|)=\mathcal{O}(b_i(\sigma_i-1))=\mathcal{O}(b_i(\sigma_i-\kappa_i))=\mathcal{O}(d_v)$.
Additionally, since steps \ref{st_start} to \ref{st_end} have to be repeated $s_i=\binom {d_i-1} {\sigma_i-1}$ times, one has
\begin{align}\label{eq_Complex}
	\mathcal{C}_P=s_i\mathcal{C}_L=\mathcal{O}\left(\binom {d_i-1} {\sigma_i-1}(n+d_v)d_v^{1.5}\right).
\end{align}


To continue, from Theorem \ref{Th_1}, one has $\sigma_i\in[n\kappa_i + 2, d_i-\kappa_i]$ and $\gamma_i=d_i-\kappa_i-\sigma_i+2$. Thus, choosing $\sigma_i=n\kappa_i + 2$ leads to larger $\gamma_i$. This means the obtained ${v}_i^\star$ uses more information from its normal neighbors, so that can potentially increase the convergence rate of the consensus-based distributed algorithm. As will be shown in the next section, such choice of $\sigma_i$ will be used to establish a  theoretical guarantee on an exponential convergence rate.
By taking $\sigma_i=n\kappa_i + 2$, equation \eqref{eq_Complex} yields
\begin{align}\label{eq_Complex2}
	\mathcal{C}_P=\mathcal{O}\left(\binom {d_i-1} {n\kappa_i + 1}(n+d_v)d_v^{1.5}\right),
\end{align}
and equation \eqref{eq_dv} yields
\begin{align}\label{eq_dv2}
	d_v&=\binom {n\kappa_i + 1} {\kappa_i}((n-1)\kappa_i + 2)+(n\kappa_i + 1).
\end{align}
Clearly, when $n$ and $\kappa_i$ are constants, the complexity $\mathcal{C}_P$ of calculating a resilient convex combination is polynomial in the number of agent's neighbors $d_i$. 
\begin{remark}\label{Rm_CP}
	Since $\mathcal{C}_P$ is exponential in $n$, our approach is more suitable for problems with low dimensions, such as the resilient multi-robot rendezvous problem with local constraints.
	 By observing \eqref{eq_dv} and \eqref{eq_Complex}, the complexity $\mathcal{C}_P$  will decrease as $\sigma_i\to (d_i-\kappa_i)$, due to the fast decay of $\binom {d_i-1} {\sigma_i-1}$ for $\sigma_i-1\ge\frac{d_i-1}{2}$. However, doing so leads to small $\gamma_i$ and can potentially decrease the convergence speed of the consensus-based algorithm. \hfill$\square$
\end{remark}

\section{Resilient Constrained Consensus}\label{Sec_CTC}

By using the  $(\gamma_i,\alpha_i)$-Resilient Convex Combination developed in the previous section, we introduce an approach that allows multi-agent systems to achieve resilient constrained consensus in the presence of Byzantine attacks. This distinguishes the paper from existing ones that are only applicable to unconstrained consensus problems \cite{Shreyas13Selected,SH17Auto,KLS20Arxiv,PRJ17NIPS,mendes2015DC,N14ICDCN,PH15IROS,PH16ICRA,PH17TOR,WD13DCG,PME08ACM,LN14ACM,XSS18NACO}.

\subsection{Consensus Algorithm Under Byzantine Attack} \label{Sec_CCUBA}
By recalling Sec. \ref{subsec_CBDA}, consider a multi-agent network with $m$ agents characterized by the time-varying network $\mathbb{G}(t)$.
In distributed consensus problems, each agent is associated with a local state $x_i\in\mathbb{R}^n$ and a constraint set $\mathcal{X}_i\subseteq \mathbb{R}^n$. Our goal is to find a common point $x^*\in\mathbb{R}^n$ such that
\begin{align} %
	x_1=&\cdots=x_m=x^*, \qquad x^*\in\bigcap_{i=1}^m \mathcal{X}_i.\tag{\ref{Csed_Csus2}}
\end{align}
Based on \cite{WR05TAC}, for constrained consensus problems, where $\exists i\in\{1,\cdots,m\}$, $\mathcal{X}_i\subsetneq\mathbb{R}^n$, the update \eqref{eq_algCC}
is effective if there exists a sequence of contiguous uniformly bounded time intervals such that in every interval, ${\mathbb{G}}(t)$ is strongly connected for at least one time step.\footnote{A network is called strongly connected if every agent is reachable through a path from every other agent in the network.} In addition, to guarantee an exponential convergence rate, the weights of the spanning graph must be strictly positive and bounded away from $0$.

Now, we assume that the network ${\mathbb{G}}(t)$ suffers a Byzantine attack as described in Sec. \ref{subsec_Byzantine},  where the number of Byzantine neighbors of each normal agent is $\kappa_i(t)$, and $\kappa_i(t)\le\bar{\kappa}$, $\forall i\in\{1,\cdots,m\}$, $t=1,2,\cdots$. The new network is characterized by ${\mathbb{G}}^{+}(t)$. To guarantee that the normal agents in ${\mathbb{G}}^{+}(t)$ are not influenced by the Byzantine agents, in the following, we introduce a way to incorporate the approach proposed in Algorithm \ref{Algorithm_Main} to ensure resilient constrained consensus.
Along with this, instead of directly using the result in \cite{WR05TAC}, i.e., relaying on strongly connected graphs, we will introduce concepts of network redundancy and constraint redundancy (cf. Assumption \ref{AS_RD}) to derive a relaxed topological condition.
	

\subsection{Main Result: Resilient Constrained Consensus}\label{Sec_RCC}

For update \eqref{eq_algCC}, substitute the $v_i(t)$ with the ${v}_i^\star(t)$ of Algorithm \ref{Algorithm_Main}.
Then,  the vector ${v}_i^\star(t)$ is a convex combination that automatically removes any information from Byzantine agents; however, as shown in Remark \ref{R_Th1_2}, ${v}_i^\star(t)$ also excludes some information from certain normal neighbors of agent $i$. In other words, substituting the $v_i(t)$ with  ${v}_i^\star(t)$ is equivalent to a classical consensus update on a graph $\widehat{\mathbb{G}}(t)$, where all Byzantine agents (and some edges from normal agents) are removed. In particular, the network $\widehat{\mathbb{G}}(t)$ is defined as follows.

\begin{definition}[Resilient Sub-graph]\label{def_hG}
	Define a graph $\widehat{\mathbb{G}}(t)$ with $\mathscr{V}(\widehat{\mathbb{G}}(t))=\mathscr{V}({\mathbb{G}}(t))$ (same agent set).\footnote{ To clarify, ${\mathbb{G}}(t)$ is the graph without Byzantine agents.} For each time instance $t$, associate the edges of $\widehat{\mathbb{G}}(t)$ with the weights  $w^{\star}_{i\ell}$, $\forall i,\ell\in\{1,\cdots,m\}$ defined in equation \eqref{eq_reconvexcomb}, which is the ($(\gamma_i,\alpha_i)$-resilient convex combination obtained by Algorithm \ref{Algorithm_Main} with parameter $\sigma_i(t)=n\kappa_i(t)+2$.
\end{definition}

Obviously, $\widehat{\mathbb{G}}(t)$ is composed of only normal agents. Furthermore, according to Theorem \ref{Th_1}, each agent in $\widehat{\mathbb{G}}(t)$ retains at least $\gamma_i(t)=d_i(t)-\kappa_i(t)-\sigma_i(t)+2$ of its incoming edges with other normal agents, with weights no less than $\alpha_i(t)$. 
Thus, by associating the edges of $\widehat{\mathbb{G}}(t)$ with the weights in \eqref{eq_reconvexcomb}, $\widehat{\mathbb{G}}(t)$ must be an edge-induced sub-graph of ${\mathbb{G}}(t)$.
We choose $\sigma_i(t)=n\kappa_i(t)+2$ in order to  maximize the number of retained edges in $\widehat{\mathbb{G}}(t)$. In this case, $\gamma_i(t)=d_i(t)-(n+1)\kappa_i(t)$.

Now, to see whether the consensus update running on $\widehat{\mathbb{G}}(t)$ is able to solve the problem \eqref{Csed_Csus2},  according to \cite{PW13TAC}, we need to check whether there exists a sequence of contiguous uniformly bounded time intervals such that the union of network $\widehat{\mathbb{G}}(t)$ across each interval is strongly connected.
However, since the edge elimination for each agent in ${\mathbb{G}}(t)$ is a function of the unknown (and arbitrary) actions of Byzantine agents,
it can be difficult to draw any conclusion on the strong-connectivity of the union of   $\widehat{\mathbb{G}}(t)$, based on the union of ${\mathbb{G}}(t)$. 
Actually,  even if the ${\mathbb{G}}(t)$ is fully connected, the obtained $\widehat{\mathbb{G}}(t)$ may still be non-strongly connected.

Here, we provide a toy example to show why strong-connectivity is critical. Consider a fully connected network ${\mathbb{G}}(t)$ with $m=4$ agents, where each agent possesses a local constraint $\mathcal{X}_i$. Suppose for all $i\in\{1,2,3,4\}$, $\kappa_i=1$ and $n=1$. Thus, $\sigma_i=n\kappa_i(t)+2=3$, $d_i=m+\kappa_i=5$ and $\gamma_i=d_i-\kappa_i-\sigma_i+2=5-1-3+2=3$.
This guarantees that each agent has three neighbors (including itself), i.e., each agent excludes one of its neighbors from a fully connected graph. Then as shown in Fig. \ref{Fig_weak_4}, if all agents 1-3 exclude agent 4 from their neighbor sets, then the obtained $\widehat{\mathbb{G}}(t)$ cannot be strongly connected, because there is no path to deliver information from agent $4$ to agents $1,2,3$.
\noindent\begin{figure}[t]
	\vspace{-0.1cm}
	\centering
	\includegraphics[width=5.5 cm]{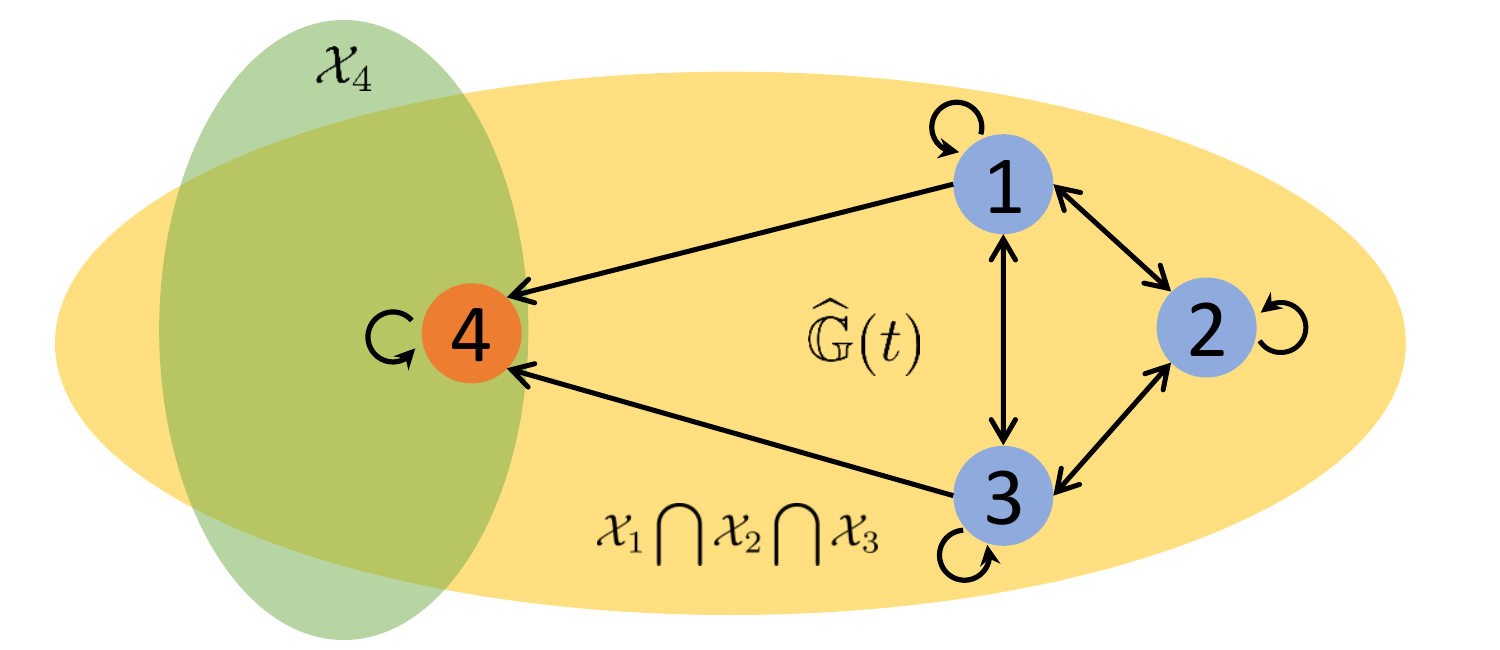}
	\caption{A resilient sub-graph of 4 agents with $\gamma_i=3$. Yellow shade denotes the constraint of agents 1-3, green shade denotes the constraint of agent 4.}
	\label{Fig_weak_4}
\end{figure}
Note that, for the $\widehat{\mathbb{G}}(t)$ in Fig. \ref{Fig_weak_4}, agents 1-3 are strongly connected, and by running update \eqref{eq_algCC}, they are able to reach a consensus within the intersection of $\bigcap_{i=1}^3\mathcal{X}_i$. However, since agents 1-3 are not able to receive any information from agent 4, they will not move towards agent 4 and, therefore, have no awareness of the constraint $\mathcal{X}_4$. On the other hand for agent 4, even though it can receive information from agents 1 and 3, the existence of constraint $\mathcal{X}_4$ prevents agent 4 from moving towards agents 1-3, thus, a consensus cannot be reached, even though there are no Byzantine agents in the network.

As validated by the example in Fig. \ref{Fig_weak_4}, if one directly replaces the $v_i(t)$ in \eqref{eq_algCC} with the ${v}_i^\star$ in Algorithm \ref{Algorithm_Main}, the underlying network $\widehat{\mathbb{G}}(t)$ may not be sufficient for the agents in the network to reach resilient constrained consensus. 
This phenomenon necessitates the following assumption, where we introduce redundancies for both the network and the constraint sets.
\begin{assumption}[Redundancy Conditions]\label{AS_RD} We assume there exists $\varphi\in\mathbb{Z}_{+}$ such that the following conditions hold.
	\begin{enumerate}[label=\textbf{\alph*}.]
		\item $\textit{[Network Redundancy]}$ {By Theorem \ref{Th_1}, one has $\alpha_i(t)=\frac{1}{(s_i(t)+1)(n+1)}$. Define $\underline{\alpha}=\min \{\alpha_i(t)\}$, $\forall t\in\{1,2,\cdots\}$ and $\forall i\in\{1,2,\cdots,m\}$.  Such $\underline{\alpha}\in\mathbb{R}_+$ must exist and is bounded away from zero due to the boundedness of $s_i(t)$.} Now, consider the network $\widehat{\mathbb{G}}(t)$ defined in Definition \ref{def_hG}. Suppose $\forall t=1,2,\cdots$, $\widehat{\mathbb{G}}(t)$ is spanned by a rooted directed graph\footnote{A spanning rooted graph contains a least one agent that can reach every other agent in the network through a path.} $\mathbb{N}(t)$,  where $\mathscr{V}(\mathbb{N}(t))=\mathscr{V}(\widehat{\mathbb{G}}(t))$ and $\mathbb{N}(t)$ has the edges of $\widehat{\mathbb{G}}(t)$ with weights no less than $\underline{\alpha}$. 
		Furthermore, let $\mathbb{S}(t)$ be the root strongly connected component of $\mathbb{N}(t)$, and suppose there exists an infinite sequence of contiguous uniformly bounded time intervals such that in every interval, $|\mathscr{V}(\mathbb{S}(t))|\ge \varphi$ for at least one time step.

		\item $\textit{[Constraints Redundancy]}$
		Let $\mathcal{V}=\{1,\cdots,m\}$ denote the agent set of network $\mathbb{G}(t)$. Suppose for all $\bar{\mathcal{V}}\subset \mathcal{V}$ with $|\bar{\mathcal{V}}|\ge \varphi$, there holds
		\begin{align}\label{eq_subsetX}
			\bigcap_{i\in\bar{\mathcal{V}}}\mathcal{X}_i=\bigcap_{i=1}^m \mathcal{X}_i.
		\end{align}
	\end{enumerate}
\end{assumption}

The proposed assumption leads to the following result.

\begin{theorem}\longthmtitle{Resilient Constrained Consensus}  \label{Th_MCC}
	Consider a network ${\mathbb{G}}(t)$ of normal agents. Suppose ${\mathbb{G}}(t)$ experiences a Byzantine attack such that each normal agent in ${\mathbb{G}}(t)$ has ${\kappa}_i(t)$ malicious neighbors. Suppose Assumptions \ref{AS_RD}{a} and {b} hold.
	Then, by replacing the $v_i(t)$ in \eqref{eq_algCC} by the ${v}_i^\star$ in Algorithm \ref{Algorithm_Main} with $\sigma_i(t)=n\kappa_i(t)+2$, update \eqref{eq_algCC} will drive the states in all agents of ${\mathbb{G}}(t)$ to converge exponentially fast to a common state which satisfies the constrained consensus \eqref{Csed_Csus2}.
\end{theorem}

\noindent{\bf Proof of Theorem \ref{Th_MCC}:}\\
Here, 
we prove Theorem \ref{Th_MCC} by referring to some existing results in \cite{J63AMS,JH77MAA,KX12RM,PW13TAC,CW-ZJ-RW:15}. 
Let $\widehat{W}(t)$ be the (row stochastic) weight matrix corresponding to the network $\widehat{\mathbb{G}}(t)$.  Define matrix $\widehat{U}(T)$ as:
\begin{align}\label{eq_product_of_W}
	\widehat{U}(T)\triangleq \widehat{W}(T)\cdot \widehat{W}(T-1)\cdots \widehat{W}(1)\cdot \widehat{W}(0).
\end{align}
Since all $\widehat{W}(t)$ are row stochastic and all $\widehat{\mathbb{G}}(t)$ are spanned by a rooted directed graph with edge weights no less than $\underline{\alpha}$ (Assumption  \ref{AS_RD}{a}), then according to \cite{J63AMS,JH77MAA,KX12RM}, as $T\to\infty$, the $\widehat{U}(T)$ in equation \eqref{eq_product_of_W} must converge exponentially fast to a constant matrix $\widehat{U}^*$ such that all rows of $\widehat{U}^*$ are identical, i.e. $\widehat{U}^*=\bm{1}_m \widehat{u}^*$, $\widehat{u}^*\in\mathbb{R}^{1\times m}$.
In addition, from Assumption \ref{AS_RD}{a}, $\widehat{\mathbb{G}}(t)$ contains a component $\mathbb{S}(t)$ with at least $\varphi$ agents (containing the root of $\mathbb{N}(t)$), and for $t=1,2,\cdots$, $\mathbb{S}(t)$ is strongly connected infinitely often. {Furthermore, since $\mathbb{N}(t)$ has a finite number of agents, by Infinite Pigeonhole Principle \cite{AM94Combi}, there must exist a fixed set of $\varphi$ agents (containing the root of $\mathbb{N}(t)$) such that the graph induced by these agents is strongly connected infinitely often. 
Corresponding to this set of agents, at least $\varphi$ columns of $\widehat{U}^*$ are strictly positive~\cite{CXL16TAC}. 
Furthermore, since $\widehat{U}^*=\bm{1}_m \widehat{u}^*$, at least $\varphi$ entries of $\widehat{u}^*$ must be strictly positive.}
For $i\in\{1,\cdots,m\}$, define
\begin{align}
	e_i(t)=\left[\mathcal{P}_i(v_i(t))-v_i(t)\right].
\end{align}
Accordingly, update \eqref{eq_algCC} can be rewritten as 
\begin{align}
	x_i(t+1)=v_i(t)+e_i(t),  \label{eq_rwalgCC}
\end{align}
where $v_i(t)={v}_i^\star=\sum_{\ell\in{\widehat{\mathcal{N}}}_i}w_{i\ell}(t)x_\ell(t)$ and the $w_{i\ell}(t)$ is the $i,\ell$th entry of the adjacency matrix $\widehat{W}(t)$. 

To proceed, recall our above proof  that the $\widehat{U}(T)$ defined in \eqref{eq_product_of_W} converges exponentially fast to the matrix $\bm{1}_m\widehat{u}^*$, and the entries of $\widehat{u}^*$ associated with $\mathscr{V}(\mathbb{S}(t))$ must be strictly positive. Then by utilizing Lemma 8 in \cite{PW13TAC}\footnote{The main result in \cite{PW13TAC} considers the convergence of states to the intersection of constraints in all agents, which requires \textit{all} entries in $\widehat{u}^*$ to be strictly positive. Here, we utilize Lemma 8 in \cite{PW13TAC} to prove the convergences of states to the intersection of constraints possessed by agents in $\mathbb{S}(t)$, since only the corresponding entries in $\widehat{u}^*$ are strictly positive.}, for $i\in\mathcal{I}$, there holds $\|e_i(t)\|\to 0$ as $t\to \infty$. This further implies that for $i\in\mathscr{V}(\mathbb{S}(t))$, the states $x_i(t)$ will reach a consensus at a point $x^*$, such that $\mathcal{P}_i(x^*)=x^*$, $i\in\mathcal{I}$. Finally, by Assumption \ref{AS_RD}b, since $|\mathscr{V}(\mathbb{S}(t))|\ge\varphi$, one has $x^*\in\bigcap_{i\in\mathscr{V}(\mathbb{S}(t))} \mathcal{X}_i=\bigcap_{i=1}^m \mathcal{X}_i$.

Now, to complete the proof of Theorem \ref{Th_MCC}, we only need to show that for agents $j\notin \mathscr{V}(\mathbb{S}(t))$, the states $x_j(t)$ also converge to $x^*$. Note that
$x^*\in\bigcap_{i=1}^m \mathcal{X}_i\subset \bigcap_{j\notin\mathscr{V}(\mathbb{S}(t))} \mathcal{X}_j$.
Thus, for agents $j\notin \mathscr{V}(\mathbb{S}(t))$, their dynamics are dominated by a leader-follower consensus process within the domain of $\bigcap_{j\notin\mathcal{I}} \mathcal{X}_j$, where any agent in $\mathscr{V}(\mathbb{S}(t))$ can be considered as the leader. Since from Assumption \ref{AS_RD}{a}, $\mathbb{N}(t)$ is a rooted graph and the root is in $\mathscr{V}(\mathbb{S}(t))$, one has $x_j(t)$, $j\notin \mathscr{V}(\mathbb{S}(t))$ also converge exponentially to $x^*$ \cite{CW-ZJ-RW:15}. This completes the proof. \hfill\qed

\begin{remark}\longthmtitle{Justification of Assumption \ref{AS_RD}} \label{CO_SGCC}
While the example in Fig. \ref{Fig_weak_4} explains why Assumption \ref{AS_RD} is required, the above proof demonstrates that this assumption is sufficient for using Algorithm \ref{Algorithm_Main} to achieve resilient constraint consensus. Specifically, without the network redundancy, the states in all agents may not reach a consensus; without the constraint redundancy, the consensus point may not satisfy the constraints in all agents.
	
	Obviously, for Assumption \ref{AS_RD}{b},  one can guarantee that equation \eqref{eq_subsetX} holds by introducing overlaps on the local constraint $\mathcal{X}_i$ of each agent. However, for Assumption \ref{AS_RD}{a}, since $\widehat{\mathbb{G}}(t)$ is an edge-induced sub-graph of ${\mathbb{G}}(t)$, it not straightforward to see what ${\mathbb{G}}(t)$ can lead to a $\widehat{\mathbb{G}}(t)$ that satisfies the network redundancy conditions in Assumption \ref{AS_RD}{a}. Motivated by this, in the following,  we propose a sufficient condition for network redundancy. Based on this, one can easily design a network ${\mathbb{G}}(t)$, which always guarantees that Assumption \ref{AS_RD}{a} holds for certain $\varphi$.   \hfill$\square$
\end{remark}

\medskip
\begin{definition}\longthmtitle{$r$-reachable set \cite{Shreyas13Selected}}\label{def_rreach}
	Given a directed graph $\mathbb{G}$ and a nonempty subset $\mathcal{A}_s$ of its agents $\mathscr{V}(\mathbb{G})$, we say $\mathcal{A}_s$ is an r-reachable set if $\exists i \in \mathcal{A}_s$ such that $|\mathcal{N}_i \setminus \mathcal{A}_s| \ge r$.
\end{definition}

\begin{corollary}\label{CO_A1C}\longthmtitle{Sufficient condition for network redundancy}
	 Suppose for all $ t=1,2,\cdots$, the network ${\mathbb{G}(t)}$ has a fully connected sub-graph\footnote{A fully connected component is a complete graph, i.e., all agents in the component have incoming edges from all other agents in that component.} $\mathbb{F}$ (time-invariant) whose number of agents satisfies $|\mathscr{V}(\mathbb{F})|=f\ge2n\bar{\kappa}+1$.
	 Furthermore, we assume that any subset of $\mathscr{V}({\mathbb{G}(t)})\setminus \mathscr{V}(\mathbb{F})$ (set of agents that are in ${\mathbb{G}(t)}$ but not in $\mathbb{F}$) is $(n\bar{\kappa}+1)$-reachable in ${\mathbb{G}(t)}$.
	 Now if the network ${\mathbb{G}(t)}$ is under Byzantine attack and each agent adds ${\kappa}_i(t)$ malicious neighbors in each time step, such that $\kappa_i(t)\le\bar{\kappa}$, we let each agent in the system run Algorithm \ref{Algorithm_Main} by choosing $\sigma_i(t)=n\bar\kappa_i+2$. Then Assumption \ref{AS_RD}{a} must hold for $\varphi=f-n\bar{\kappa}$.
\end{corollary}

\begin{proof}
Recall from Definition \ref{def_hG} that $\widehat{\mathbb{G}}(t)$  is the network associated with the weight coefficients of \eqref{eq_reconvexcomb} obtained in Algorithm \ref{Algorithm_Main}. According to Theorem \ref{Th_1}, each agent $i$ of $\widehat{\mathbb{G}}(t)$ must have at least $\gamma_i(t)=d_i(t)-\kappa_i(t)-\sigma_i(t)+2$ neighbors (including itself) with weights no less than  $\alpha_i(t)= \frac{1}{(s_i(t)+1)(n+1)}$. 
Here, $d_i(t)=m_i(t)+\kappa_i(t)$ with $d_i(t)=|\mathcal{N}_i^+(t)|$ and $m_i(t)=|\mathcal{N}_i(t)|$. 
Thus, we can easily derive $\gamma_i(t)=m_i(t)-n{\kappa}_i(t)\ge m_i(t)-n\bar{\kappa}$.
To continue, let $\underline{\alpha}=\min~\{\alpha_i(t)~|~i=1,\cdots,m;~ t=1,2\cdots\}$. Followed by Assumption \ref{AS_RD}{a}, for each time-step, we define the edge-induced sub-graph $\mathbb{N}(t)$ of $\widehat{\mathbb{G}}(t)$, which removes all the edges of $\widehat{\mathbb{G}}(t)$ with weights less than $\underline{\alpha}$. Since $\alpha_i(t)\ge\underline{\alpha}$ and $\gamma_i(t)\ge m_i(t)-n\bar{\kappa}$,  each agent in $\mathbb{N}(t)$ loses at most $n\bar{\kappa}$ incoming edges compared with that of the original graph ${\mathbb{G}(t)}$.

\begin{figure}[t]
 	\vspace{-0.1cm}
 	\centering
 	\includegraphics[width=8 cm]{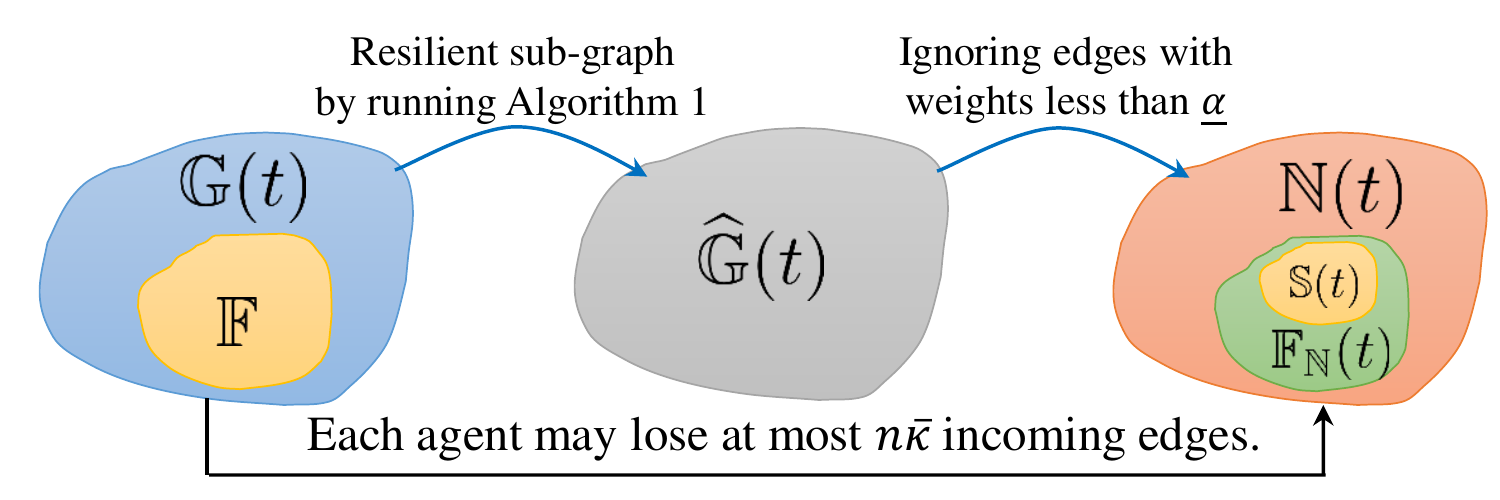}
 	\caption{A demonstration of $\mathbb{G}(t)$, $\widehat{\mathbb{G}}(t)$, $\mathbb{N}(t)$, $\mathbb{F}$, and $\mathbb{F}_{\mathbb{N}}(t)$, where $\mathbb{G}(t)$, $\widehat{\mathbb{G}}(t)$ and $\mathbb{N}(t)$ share the same agent set. $\mathbb{F}$ and $\mathbb{F}_{\mathbb{N}}(t)$ share the same agent set.}
 	\label{Fig_C1_2}
\end{figure}

We start by proving the second statement of Assumption \ref{AS_RD}{a}. This will be done by ensuring a sufficient condition such that for all $t$, a sub-graph of $\mathbb{N}(t)$ always contains a strongly connected component $\mathbb{S}(t)$ with at least $\varphi=f-n\bar{\kappa}$ agents. 
Since  ${\mathbb{G}(t)}$ has a fully connected component $\mathbb{F}$, let $\mathbb{F}_{\mathbb{N}}(t)$ denote the agent-induced sub-graph of $\mathbb{N}(t)$ such that $\mathscr{V}(\mathbb{F}_{\mathbb{N}}(t))=\mathscr{V}(\mathbb{F})$. 
Then $|\mathscr{V}(\mathbb{F}_{\mathbb{N}}(t))|=|\mathscr{V}(\mathbb{F})|=f$, and each agent in $\mathbb{F}_{\mathbb{N}}(t)$ loses at most $n\bar{\kappa}$ incoming edges from $\mathbb{F}$ (this property follows from the relation between $\mathbb{G}(t)$ and $\mathbb{N}(t)$, as visualized in Fig. \ref{Fig_C1_2}). 
As a consequence, every agent in  $\mathbb{F}_{\mathbb{N}}(t)$ has at least $f-n\bar{\kappa}=\varphi$ incoming edges with the agents in $\mathbb{F}_{\mathbb{N}}(t)$ (including self-loops). 
Now, for time step $t$, suppose $\mathbb{F}_{\mathbb{N}}(t)$ has $N(t)$ maximal strongly connected components (disjoint), denoted by $\mathbb{S}_j(t)$, $j\in\{1,\cdots,N(t)\}$. Clearly, there exists at least one root strongly connected component $\mathbb{S}(t)\in\{\mathbb{S}_j(t)\}$ that does not have incoming edges from other components in $\mathbb{F}_{\mathbb{N}}(t)$ \cite{WDB:2001}. Because otherwise there will be loops among $\mathbb{S}_j(t)$, leading to a larger component which contradicts with its definition. 
Based on this, recall the fact that every agent in $\mathbb{F}_{\mathbb{N}}(t)$, as well as every agent in $\mathbb{S}(t)$, has at least $f-n\bar{\kappa}=\varphi$ incoming edges from the agents in $\mathbb{F}_{\mathbb{N}}(t)$. Thus, all incoming edges of agents in $\mathbb{S}(t)$ must come from $\mathbb{S}(t)$ itself (including self-loops). Consequently, $|\mathscr{V}(\mathbb{S}(t))|\ge \varphi$.

In the following, we establish the remaining statements of \ref{AS_RD}\textbf{a}, namely, $\mathbb{N}(t)$ is a rooted (connected) graph and one root of the graph is in $\mathbb{S}(t)$. We start by showing $\mathbb{F}_{\mathbb{N}}(t)$ is a rooted graph. Note that from the above derivation, $\mathbb{S}(t)$ is a strongly connected component with $|\mathscr{V}(\mathbb{S}(t))|\ge \varphi=f-n\bar{\kappa}$. Thus, $|\mathscr{V}(\mathbb{F}_{\mathbb{N}}(t))|-|\mathscr{V}(\mathbb{S}(t))|\le n\bar{\kappa}$, meaning that all agents in $\mathbb{F}_{\mathbb{N}}(t)$ but not in $\mathbb{S}(t)$ must have at least one incoming edge from $\mathbb{S}(t)$. Therefore, $\mathbb{F}_{\mathbb{N}}(t)$ is a rooted graph and one root of the graph is in $\mathbb{S}(t)$.
Now, we show that $\mathbb{N}(t)$ is a rooted graph and one of its roots is in $\mathbb{S}(t)$. 
Recall that any subset of agents in $\mathscr{V}({\mathbb{G}(t)})\setminus \mathscr{V}(\mathbb{F})$ is $(n\bar{\kappa}+1)$-reachable in ${\mathbb{G}(t)}$. Since each agent in $\mathbb{N}(t)$ loses at most $n\bar{\kappa}$ incoming edges from $\mathbb{G}(t)$, any subset of $\mathscr{V}({\mathbb{N}(t)})\setminus \mathscr{V}(\mathbb{F}_{\mathbb{N}}(t))$ is $1$-reachable in ${\mathbb{N}(t)}$. Let us decompose $\mathbb{N}(t)$ into its strongly
connected components. If $\mathbb{N}(t)$ is not a rooted graph, there must exist at least two components, denoted by $\mathbb{P}_1$ and $\mathbb{P}_2$, that have no incoming edges from any other components. Since $\mathbb{F}_{\mathbb{N}}(t)$ is a rooted graph, any agents in $\mathbb{F}_{\mathbb{N}}(t)$ can only be associated with one of such root component, say $\mathbb{P}_1$. Then $\mathbb{P}_2$ must be composed of agents only in $\mathscr{V}({\mathbb{N}(t)})\setminus \mathscr{V}(\mathbb{F}_{\mathbb{N}}(t))$. This, however, contradicts with the fact that any subset of $\mathscr{V}({\mathbb{N}(t)})\setminus \mathscr{V}(\mathbb{F}_{\mathbb{N}}(t))$ is $1$-reachable in ${\mathbb{N}(t)}$ (every subset must has a neighbor outside the set). Thus, $\mathbb{N}(t)$ must be a rooted graph. Further since there can be no root components in $\mathscr{V}({\mathbb{N}(t)})\setminus \mathscr{V}(\mathbb{F}_{\mathbb{N}}(t))$. The root component of ${\mathbb{N}(t)}$ must be in $\mathbb{F}_{\mathbb{N}}(t)$, which is then rooted by agents in $\mathbb{S}(t)$.
This completes the proof. \end{proof}

\smallskip
Note that if Assumption \ref{AS_RD}{a} holds for $\varphi=f-n\bar{\kappa}$, it also holds for any smaller $\varphi$.

\subsection{Special Case: Resilient Unconstrained Consensus}
Note that if the agents in the system are not subject to local constraints, i.e. $\mathcal{X}_i=\mathbb{R}^n,~ \forall i=1,\cdots,m$ and $\mathcal{P}_i(v_i(t))=v_i(t)$, then update \eqref{eq_algCC} degrades to 
\begin{align}\label{eq_algCCS}
	x_i(t+1)
	=v_i(t).
\end{align}
As pointed in \cite{WR05TAC}, for unconstrained consensus, a sufficient condition for the effectiveness of update \eqref{eq_algCCS} is the existence of a uniformly bounded sequence of contiguous time intervals such that in every interval, $\widehat{\mathbb{G}}(t)$ has a rooted directed spanning tree
for at least one time step. Obviously, this connectivity condition is weaker than the one proposed in Sec. \ref{Sec_CCUBA} for the constrained consensus case.
Correspondingly, we introduce a weaker version of Assumption \ref{AS_RD} as follows.

	\begin{definition}\longthmtitle{$r$-robustness \cite{Shreyas13Selected}}
	A directed graph $\mathbb{G}$ is $r$-robust, with $r\in\mathbb{Z}_{\ge0}$, if for every pair of nonempty, disjoint subsets of $\mathbb{G}$,
	at least one of the subsets is $r$-reachable. 
\end{definition}
\begin{assumption} \label{As_UcSc}
	Consider a network $\mathbb{G}(t)$. Suppose $\mathbb{G}(t)$ is $(n\bar{\kappa}+1)$-robust. 
\end{assumption}

This assumption leads to the following lemma, which can be obtained by Lemma 7 of \cite{Shreyas13Selected} or by combining the Lemmas 3 and 4 of \cite{N14ICDCN}:

\begin{lemma}\label{Lm_UcSc}
	For any network  $\mathbb{G}(t)$ satisfying Assumption \ref{As_UcSc}, if each agent in $\mathbb{G}(t)$ removes up to $n\bar{\kappa}$ incoming edges, the obtained edge-induced network is still spanned by a rooted directed tree.
\end{lemma}

Consequently, one has the following.

\begin{theorem}  \label{Th_MUC}
	Consider a network ${\mathbb{G}}(t)$ of normal agents. For ${\mathbb{G}}(t)$, there exists a sequence of uniformly bounded contiguous time intervals such that in each interval at least one ${\mathbb{G}}(t)$ satisfies Assumption \ref{As_UcSc}. Suppose ${\mathbb{G}}(t)$ is affected by a Byzantine attack such that each normal agent in ${\mathbb{G}}(t)$ has at most $\bar{\kappa}$ Byzantine neighbors. Then by replacing the $v_i(t)$ in \eqref{eq_algCCS} with the ${v}_i^\star$ in Algorithm \ref{Algorithm_Main} with $\sigma_i(t)=n\kappa_i(t)+2$, update \eqref{eq_algCCS} will drive the states in all agents of ${\mathbb{G}}(t)$ to reach a consensus exponentially fast.
\end{theorem}

The proof of this Theorem can directly be obtained by combining the Theorem \ref{Th_1}, Lemma \ref{Lm_UcSc} in this paper and the result in \cite{WR05TAC}.

\begin{remark}
	Note that the results in this subsection can be considered as the special case of Sec. \ref{Sec_RCC}. Particular, compared with Assumption \ref{AS_RD}{a}, unconstrained consensus no longer requires a root strongly connected component such that $|\mathscr{V}(\mathbb{S}(t))|\ge \varphi$, but only $|\mathscr{V}(\mathbb{S}(t))|=1$, which reduced to a rooted spanning tree. Such condition can be guaranteed by $(n\bar{\kappa}+1)$-robustness in Assumption \ref{As_UcSc} based on Lemma \ref{Lm_UcSc}.	
	The constraints redundancy Assumption \ref{AS_RD}{b} is completely removed due to the non-existence of local constraints. 
	In addition, compared with existing results, Theorem \ref{Th_MUC} applies to the resilient unconstrained consensus for multi-dimensional state vector under time-varying networks.  If the network is time-invariant, our result matches the Theorem 1 (under Condition SC) proposed in \cite{N14ICDCN}, and the Theorem 5.1 (under Assumption 5.2) proposed in \cite{PH15IROS}. If the agents' states are scalars, the result matches the Theorem 2 of \cite{Shreyas13Selected}. \hfill$\square$
\end{remark}

\section{Simulation and Application Examples}\label{Sec_SM}
In this section, we validate the effectiveness of the proposed algorithm via both numerical simulations and an application example of safe multi-agent learning. 

\subsection{Resilient Consensus in Multi-agent Systems}
For resilient consensus, consider the following update:
\begin{align}
	x_i(t+1)=\mathcal{P}_i({v}_i^\star(t)),  \label{eq_simu}
\end{align}
where ${v}_i^\star(t)$ is a resilient convex combination that can be obtained from Algorithm \ref{Algorithm_Main} with $x_{\mathcal{N}_i^+}(t)$; $\mathcal{P}_i(\cdot)$ is a projection operator that projects any state to the local constraint $\mathcal{X}_i$. If the constraint does not exist, then $\mathcal{P}_i({v}_i^\star)={v}_i^\star$. Here, we validate resilient consensus problems for both unconstrained and constrained cases.

\subsubsection*{Unconstrained consensus}
Consider a multi-agent system composed of 8 normal agents. Let $n=2$ and $\bar{\kappa}=1$, which means that each agent holds a state $x_i\in\mathcal{R}^2$ and has at most one malicious neighbors. Suppose the communication network of the system is time-varying, but every 3 time steps, at least one graph satisfies the connectivity condition described in Assumption \ref{As_UcSc}. The following Fig. \ref{Fig_Network} is an example of the network at a certain time-step, which is constructed by the \textit{preferential-attachment} method \cite{RA02RMP}. Note that in Fig. \ref{Fig_Network}, we do not draw concrete Byzantine agents as these agents can send different misinformation to different normal agents. Specifically, the red arrows can be considered either different information for one Byzantine agent or different information for different Byzantine agents.
\noindent\begin{figure}[h!]
	\vspace{-0.1cm}
	\centering
	\includegraphics[width=7 cm]{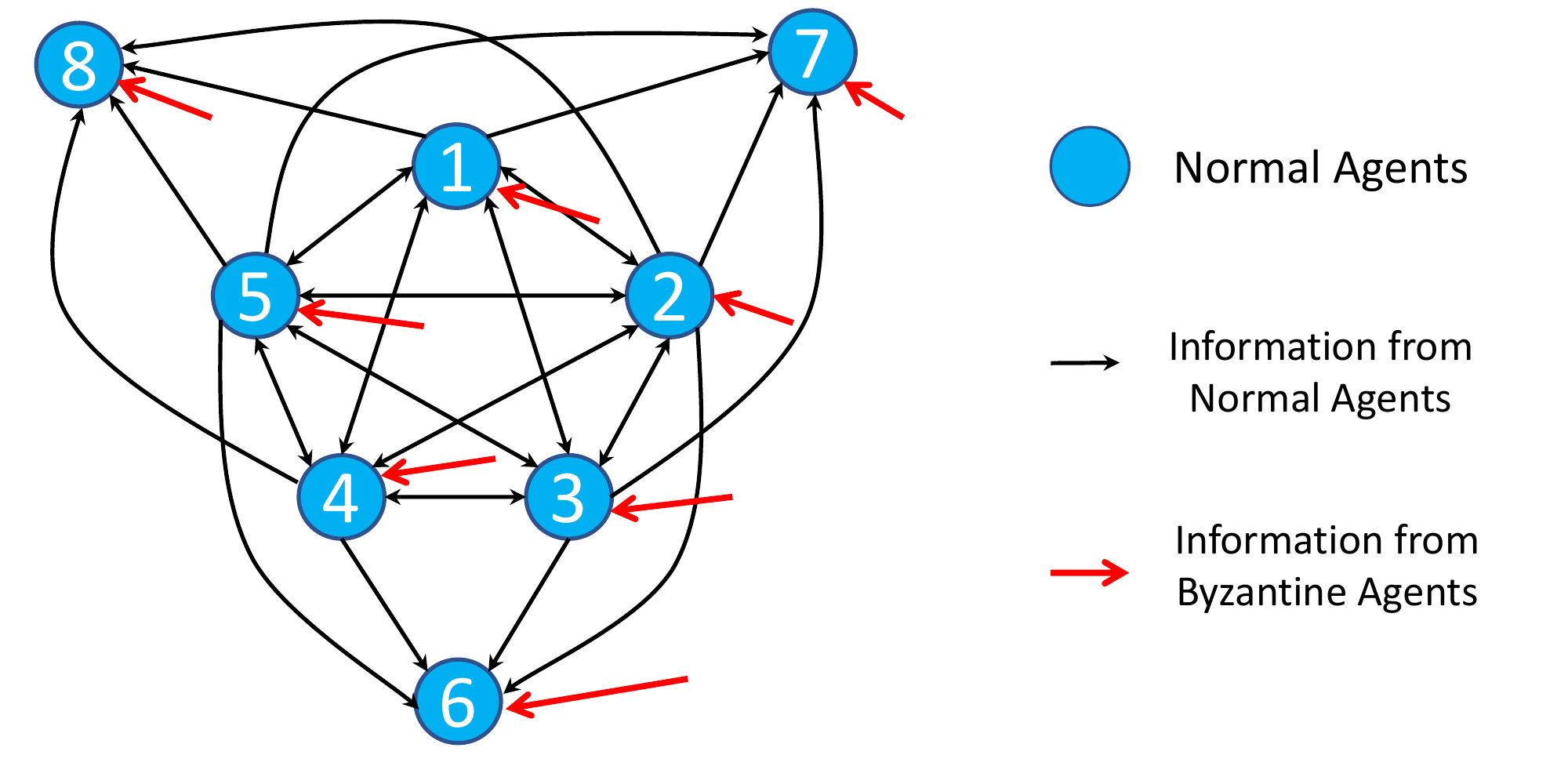}
	\caption{A network of 8 agents, under Byzantine attack.}
	\label{Fig_Network}
\end{figure}

For Fig.~\ref{Fig_Network}, each agent has $d_i=6$ (including the edge from Byzantine agents and the agent's self-loop).  According to Theorem \ref{Th_1}, by taking $\sigma_i=n\kappa_i+2$, each normal agent can locally compute a $(\gamma_i,\alpha_i)$-resilient convex combination such that for all $i=1,\cdots,8$, there holds $\gamma_i=d_i-\kappa_i-\sigma_i+2=3$. Note that the network in Fig. \ref{Fig_Network} satisfies the condition described in Assumption \ref{As_UcSc}. Although Assumption \ref{AS_RD} \textbf{a} is not required for unconstrained consensus, it is worth mentioning that Fig. \ref{Fig_Network} also satisfies the condition described in Corollary \ref{CO_A1C} with a fully connected component of $5$ agents. This leads to the hold of Assumption \ref{AS_RD} \textbf{a}. 

Initialize the normal agents' states $x_i(0)$ by vectors randomly chosen from  $[-2, 2]\times[-2, 2]$. At each times step, let the malicious agents send a state to normal agents, which is randomly chosen from $[-4, 4]\times[-5, 5]$. By running update \eqref{eq_simu}, one has the following simulation result.
\noindent\begin{figure}[h!]
	\vspace{-0.1cm}
	\centering
	\includegraphics[width=0.5\textwidth]{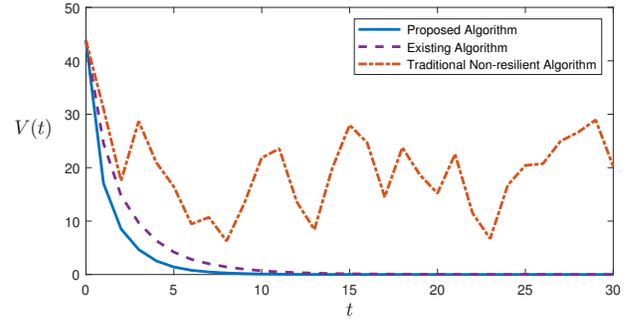}
	\caption{The algorithm proposed in this paper v.s. the result of \cite{HMNV15DC} for unconstrained consensus v.s. tradition non-resilient consensus algorithm.}
	\label{Fig_UCSimu}
\end{figure}

\noindent In the figure, $V(t)\triangleq\sum_{i=1}^m\sum_{j=1}^m \|x_i(t)-x_j(t)\|_2^2$
measures the closeness among all the normal neighbors. $V(t)=0$ if and only if a consensus is reached. It can be seen that compared with traditional non-resilient consensus update \eqref{eq_algCC}, where the states of the normal agents will be misled by their Byzantine neighbors, the proposed approach allows all the normal agents to achieve a resilient consensus. 
{Additionally, we note that the topology in Fig.~\ref{Fig_Network} also satisfies the network condition in \cite{HMNV15DC}. The algorithm in \cite{HMNV15DC} achieves resilient consensus with a slightly different convergence rate. The improved convergence rate of our algorithm might be attributed to the special use of agents' own states, as we have discussed in Remark \ref{Rm_2}.}

\subsubsection*{Constrained consensus}
It has been shown in Sec.~\ref{Sec_RCC} that a resilient constrained consensus requires both network and constraint redundancy, in this simulation, we employ a network $\mathbb{G}(t)$ of $m=30$ agents. Let $n=2$. Suppose ${\mathbb{G}}(t)$ is affected by Byzantine attack such that each normal agent in ${\mathbb{G}}(t)$ has at most $\bar{\kappa}=2$ malicious neighbors. By choosing $\varphi=11$, we construct the network such that ${\mathbb{G}}(t)$ has a fully connected component of $f=15$ agents, and any other agents of ${\mathbb{G}}(t)$ have at least $6$ neighbors in the fully connected component. This setup guarantees all conditions in Corollary \ref{CO_A1C} hold, i.e., $\varphi=f-n\bar{\kappa}$ and $f\ge2n\bar{\kappa}+1$.  Thus, the network redundancy in Assumption \ref{AS_RD} \textbf{a} must be satisfied.
To continue, for all the agents, define their local constraints as follows
\begin{align*}
	\mathcal{X}_i&= \{x|~\|x\|_2\le 1,~l_1\cdot x\ge 0\}~~\text{for}~ i=1,\cdots,10\\
	\mathcal{X}_i&= \{x|~l_1\cdot x\ge 0,~l_2\cdot x\ge 0\}~~\text{for}~ i=11,\cdots,20\\
	\mathcal{X}_i&= \{x|~l_2\cdot x\ge 0,~ \|x\|_2\le 1\}~~\text{for}~ i=21,\cdots,30
\end{align*} 
where $l_1=[1~-1]$, $l_2=[1~~0]$.
Clearly, the intersection of constraints of any $11$ agents equals to that of all the agents. Thus, the constraint redundancy in Assumption \ref{AS_RD} \textbf{b} is also satisfied for $\varphi=11$.

Under the above setup, the simulation result for this example is shown in Fig. \ref{Fig_CCSimu}.
\noindent\begin{figure}[h!]
	\vspace{-0.1cm}
	\centering
	\includegraphics[width=0.5\textwidth]{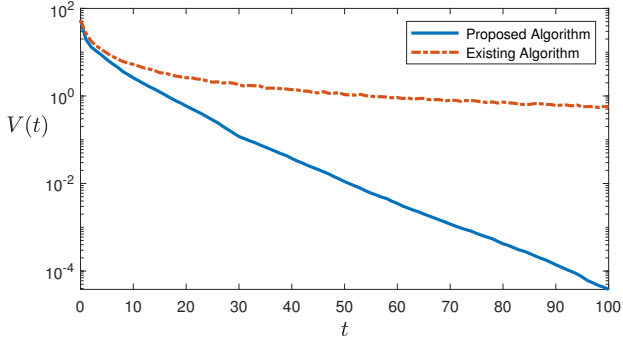}
	\caption{The algorithm proposed in this paper v.s. the result of \cite{XSS18NACO} for constrained consensus.}
	\label{Fig_CCSimu}
\end{figure}
Particularly, one has $x_i=x^*$ for all $i=1,\cdots,30$, where $x^*=[ 0.7071 ~0.7071]^\top$. Such $x^*$ satisfies all the local constraints. Furthermore, by the $log(\cdot)$ scale on $Y$ axis, one can observe that compared with the result of \cite{XSS18NACO}, the approach proposed in this paper has an exponential convergence rate. 

\subsection{Safe Multi-agent learning based on Resilient Consensus}
This subsection provides an example by employing the resilient convex combination into an application of safe multi-agent learning. 
Consensus-based multi-agent learning algorithms \cite{KZH18ICML} allow multiple agents to cooperatively handle complex machine learning tasks in a distributed manner. 
By combining the idea of consensus and machine learning techniques, the system can synthesize agents' local data samples in order to achieve a globally optimal model.
However, similar to other consensus-based algorithms, multi-agent learning algorithms are also vulnerable to malicious attacks. These attacks can either be caused by the presence of Byzantine agents or normal functioning agents with corrupted data.

Here, we consider a multi-agent value-function approximation problem, which is a fundamental component in both supervised learning algorithms \cite{A19ABC} (e.g. Regression, SVM) and reinforcement learning algorithms \cite{RA18MIT} (e.g. Deep $Q$-learning, Actor-Critic, Proximal Policy Optimization). Suppose we have the following value function $Q_{\theta}(z)$, 
\begin{align} \label{eq_MALval}
	Q_{\theta}(z)=\psi(z)^{\top}\theta 
\end{align}
where $\theta \in\mathbb{R}^n$ is the parameter vector, $\psi(\cdot)\in\mathbb{R}^n$ is the feature vector, and $z$ is the model input, which can be the labeled data in supervised learning, or the state action pairs in reinforcement learning. 
Under the multi-agent setup, given the 8-agent network shown in Fig. \ref{Fig_Network}, suppose each agent $i$ has local data sample pairs $(Q_{i,k},z_{i,k})$, $k=1, 2,\cdots$, which are subject to measurement noises. The goal is to find the optimal $\theta^\star$ parameter vector such that
\begin{align} \label{eq_MALprob}
	\min_{\theta}~\sum_{i} \sum_{k}\left\|Q_{\theta}(z_{i,k})-Q_{i,k}\right\|^2, \quad \text{s.t.}~\|\theta\|_2\le \rho
\end{align}
where the regularization constraint avoids overfitting \cite{PHS20ACC}. 

During the communication process, malicious parameter information is received by the normal agents through the red arrows in Fig. \ref{Fig_Network}. Thus, similar to Fig. \ref{Fig_UCSimu}, traditional multi-agent learning algorithms cannot solve \eqref{eq_MALprob} in a resilient way. Motivated by this, we introduce the following algorithm 
\begin{align}
	\theta_i(t+1)=\mathcal{P}_i\left({v}_i^\star(t)-\eta(t)g_{i,k}(\theta_i(t))   \right),  \label{eq_MAL}
\end{align}
which replace the consensus vector in traditional multi-agent learning algorithms \cite{AAP10TAC} with a resilient convex combination developed in this paper. Particularly, in \eqref{eq_MAL}, ${v}_i^\star(t)$ is the resilient convex combination of $\theta_i(t)$ computed by Algorithm \ref{Algorithm_Main} with $\sigma_i=n\kappa_i+2$; $g_{i,k}(\theta_i(t))$ is the stochastic gradient corresponding to a random sample pair $(Q_{i,k},z_{i,k})$ of agent $i$:
\begin{align*}
	g_{i,k}(\theta_i(t))&=\nabla_{\theta_i}\left\|Q_{\theta_i}(z_{i,k})-Q_{i,k}\right\|^2 \\
		&=\left(Q_{\theta_i(t)}(z_{i,k})-Q_{i,k}\right)\psi(z_{i,k}),
\end{align*}
the projection operator $\mathcal{P}_i(\theta)$ ensures the satisfactory of regularization constraints:
\begin{align*}
	\mathcal{P}_i(\theta)=\underset{\tilde\theta:\ \|\tilde\theta\|\le \rho}{\arg\min}~\|\tilde\theta-\theta\|,
\end{align*}
and $\eta(t)=\frac{1}{t}$ is a diminishing learning rate that is commonly used in stochastic gradient descent \cite{AAP10TAC}. 
By defining 
$$W(t)\triangleq\sum_{i=1}^m\|\theta_i(t)-\theta^*\|_2^2$$
and running update \eqref{eq_MAL}, one obtains the result shown in \eqref{Fig_MAL}. One can observe that in the presence of Byzantine attacks, the system can perform safe multi-agent learning, i.e., the $\theta_i(t)$ for all normal  agents converge to the desired parameter $\theta^*$.
\noindent\begin{figure}[h!]
	\vspace{-0.1cm}
	\centering
	\includegraphics[width=0.5\textwidth]{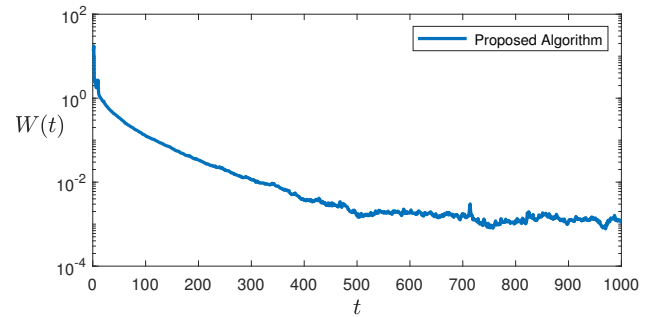}
	\caption{Multi-agent machine learning based on the resilient constrained consensus approach proposed in this paper.}
	\label{Fig_MAL}
\end{figure}

\section{Conclusion}\label{Sec_CL}
This paper aims to achieve resilient constrained consensus for multi-agent systems in the presence of Byzantine attacks. 
{We started by introducing the concept of $(\gamma_i,\alpha_i)$-resilient convex combination and proposed a new approach (Algorithm \ref{Algorithm_Main}) to compute it.} We showed that the normal agents in a multi-agent system can use their local information to automatically isolate the impact of their malicious neighbors. 
Given a fixed state dimension and an upper bound for agents' malicious neighbors, the computational complexity of the algorithm is polynomial in the network degree of each agent. Under sufficient redundancy conditions on the network topology (Assumption \ref{AS_RD} \textbf{a}) and constraints (Assumption \ref{AS_RD} \textbf{b}), the proposed algorithm can be employed to achieve resilient constrained distributed consensus with exponential convergence rate. Since the network redundancy condition (Assumption \ref{AS_RD} \textbf{a}) is difficult to directly verify, we proposed a sufficient condition (Corollary \ref{CO_A1C}), which offers an easy way to design a network satisfying the network redundancy condition. 
We have validated the correctness of all the proposed results by theoretical analysis.
We also illustrated the effectiveness of the proposed algorithms by numerical simulations, with applications in un-constrained consensus, constrained consensus, and multi-agent safe multi-agent learning.
Our future work includes further simplifying the computation complexity of the proposed algorithm and studying the trade-off between resilience and computational complexity.





%
%
%
%
%
%
%
%
%
%

\section*{Appendix}
\subsection*{Proof of Lemma \ref{LM_ERW}}
In order to prove Lemma \ref{LM_ERW}, we first introduce the following lemmas.
\begin{lemma}\label{L_PI}
	If $x_i\notin\mathcal{H}(x_{\widehat{\mathcal{S}}_{ij}})$,  then
	\begin{align}\label{HXomega}
		\bigcap_{k\in\Omega(n)}\mathcal{H}(x_{\mathcal{B}_{ijk}})\setminus \{x_i\}\neq \emptyset,
	\end{align}
	where $\Omega(n)$ is any subset of $\{1,\cdots,b_i\}$ with $|\Omega(n)|=n$.
\end{lemma}

\begin{proof}
By definition, $\forall k\in\{1,\cdots,b_i\}$, one has $x_i\in x_{\mathcal{B}_{ijk}}$.
Thus, \eqref{HXomega} holds if there exists another common state, say $x_\ell\neq x_i$, such that $x_\ell\in x_{\mathcal{B}_{ijk}}$, $\forall k\in\Omega(n)$. In the following, we show such $x_\ell$ always exists by contradiction. 

Suppose $x_i$ is the only state shared by sets $x_{\mathcal{B}_{ijk}}$, $\forall k\in\Omega(n)$.  {Since $|\Omega(n)|=n$ and  $\mathcal{B}_{ijk}\subset{\mathcal{S}}_{ij}$, the overall number of states in sets $x_{\mathcal{B}_{ijk}}$ must satisfy
\begin{align}\label{eq_BSne}
	\sum_{k\in\Omega(n)}|x_{\mathcal{B}_{ijk}}|\le n+(n-1)(|{\mathcal{S}}_{ij}|-1).
\end{align}
In the right hand side of \eqref{eq_BSne}, the first term indicates that $x_i$ appear $n$ times as elements in sets $x_{\mathcal{B}_{ijk}}$, $\forall k\in\Omega(n)$. The second term indicates that any states other than $x_i$ in $x_{\mathcal{S}_{ij}}$ may appear at most $n-1$ times as elements in sets $x_{\mathcal{B}_{ijk}}$, $\forall k\in\Omega(n)$.}
To continue, recall that $|\Omega(n)|=n$,
$|x_{\mathcal{B}_{ijk}}|=|\mathcal{B}_{ijk}|=\sigma_i-\kappa_i$, $|{\mathcal{S}}_{ij}|=\sigma_i$, and $\sigma_i\ge (n\kappa_i + 2)$, then
\begin{align}\label{eq_BSne2}
	&\sum_{k\in\Omega(n)}|x_{\mathcal{B}_{ijk}}|- n+(n-1)(|{\mathcal{S}}_{ij}|-1)\nonumber\\	=&~n(\sigma_i-\kappa_i)-(n-1)\sigma_i\nonumber\\
	=&~\sigma_i-n\kappa_i\nonumber\\
	\ge&~ n\kappa_i + 2-n\kappa_i=2.
\end{align}
Clearly, \eqref{eq_BSne} contradicts with \eqref{eq_BSne2}. Thus, there are at least two common states shared by sets $x_{\mathcal{B}_{ijk}}$, $\forall k\in\Omega(n)$. Therefore, \eqref{HXomega} must hold.
\end{proof}

\smallskip
\begin{lemma}\label{L_Bij}
	If $x_i\notin\mathcal{H}(x_{\hat{\mathcal{S}}_{ij}})$,  then
	\begin{align}\label{eq_Bijknonept}
		\left(\bigcap_{k=1}^{b_i}\mathcal{H}(x_{\mathcal{B}_{ijk}})\right)\setminus \{x_i\}\neq \emptyset.
	\end{align}
\end{lemma}

\begin{proof}
Based on Lemma \ref{L_PI}, we prove Lemma \ref{L_Bij} by induction. 
Define subsets $\Omega(q) \subset \{1,\cdots,b_i\}$, where $|\Omega(q)|=q$ and $q$ is an  integer such that  $n\le q\le b_i$. By Lemma \ref{L_PI}, when $q=n$, one has 
\begin{align}\label{HXomegab}
	\bigcap_{k\in\Omega(q)}\mathcal{H}(x_{\mathcal{B}_{ijk}})\setminus \{x_i\}\neq \emptyset.
\end{align}
{Now, given any $\Omega(q+1)$, define $\Omega_h(q)=\Omega(q+1)\setminus \{h\}$ with $h\in\Omega(q+1)$.
Since $|\Omega_h(q)|=q$, \eqref{HXomegab} yields
\begin{align}\label{HXomegabh}
	\bigcap_{k\in\Omega_h(q)}\mathcal{H}(x_{\mathcal{B}_{ijk}})\setminus \{x_i\}\neq \emptyset,\quad \forall h\in\Omega(q+1).
\end{align}
Consequently, $\forall h\in\Omega(q+1)$, there must exist points $y_h$, such that 
\begin{align}\label{defy_h}
	y_h\in\!\!\! \bigcap_{k\in\Omega_h(q)}\!\!\!\mathcal{H}(x_{\mathcal{B}_{ijk}})\setminus \{x_i\}.
\end{align}

In the following, we use contradiction to prove
\begin{align}\label{HXomegabq+1}
	\bigcap_{h\in\Omega(q+1)}\mathcal{H}(x_{\mathcal{B}_{ijh}})\setminus \{x_i\}\neq \emptyset.
\end{align}
We first assume \eqref{HXomegabq+1} does not hold, i.e., the set in \eqref{HXomegabq+1} is an empty set. Since $\Omega(q+1)=\Omega_h(q)\bigcup \{h\}$, this assumption leads to, $\forall {h}\in\Omega(q+1)$, 
\begin{align}
	\left(\bigcap_{k\in\Omega_h(q)}\!\!\!\mathcal{H}(x_{\mathcal{B}_{ijk}})\right)\bigcap \mathcal{H}(x_{\mathcal{B}_{ijh}})\setminus \{x_i\}= \emptyset.
\end{align}
Together with \eqref{defy_h}, one has $\forall {h}\in\Omega(q+1)$, 
\begin{subequations}
\begin{align}\label{defy_hh1}
	y_{{h}}&\notin\mathcal{H}(x_{\mathcal{B}_{ij{h}}})\\
		y_h&\in \mathcal{H}(x_{\mathcal{B}_{ij\widehat{h}}}),\quad \forall \widehat{h} \in\Omega(q+1),~ \widehat{h}\neq h.	\label{defy_hh2}
\end{align}
\end{subequations}
To continue, construct a set $\mathcal{Z}=\{y_h,~{h}\in\Omega(q+1)\}\bigcup\{x_i\}$ with $|\mathcal{Z}|= q+2$. Since $q\ge n$, one has $|\mathcal{Z}|\ge n+2$. Then from Radon's Theorem \cite{Tverberg66}, $\mathcal{Z}$ can always be partitioned into two sets $\mathcal{Z}_1$ and $\mathcal{Z}_2$, whose convex hulls intersect, i.e., 
\begin{align}\label{eq_defystar}
	y^*\in\mathcal{H}(\mathcal{Z}_1)\bigcap \mathcal{H}(\mathcal{Z}_2).
\end{align} 
Since $\mathcal{Z}_1$ and $\mathcal{Z}_2$ are disjoint, 
for any ${h}\in\Omega(q+1)$, $y_{{h}}$ only exists in one of them. Without loss of generality, let $\mathcal{Z}_e(h)\in\{\mathcal{Z}_1,\mathcal{Z}_2\}$  represent the subset such that $y_{{h}}\notin \mathcal{Z}_e(h)$. 
Then, based on the definition of $\mathcal{Z}$, all elements in $\mathcal{Z}_e(h)$ are either $y_{\widehat{h}}$ with $\widehat{h}\neq h$ or $x_i$. Furthermore, we have $y_{\widehat{h}}\in\mathcal{H}(x_{\mathcal{B}_{ij{h}}})$, $\widehat{h}\neq h$ and $x_i\in\mathcal{H}(x_{\mathcal{B}_{ij{h}}})$, where the first statement is derived from \eqref{defy_hh2} by exchanging the indices of $h$ and $\widehat{h}$; the second statement is by the definition of the set $x_{\mathcal{B}_{ij{h}}}$. Consequently, $\mathcal{H}(\mathcal{Z}_e(h))\subset\mathcal{H}(x_{\mathcal{B}_{ij{h}}})$. This, together with $\mathcal{Z}_e(h)\in\{\mathcal{Z}_1,\mathcal{Z}_2\}$ and \eqref{eq_defystar}, yields
\begin{align}\label{ystarinH}
	y^*\in\mathcal{H}(x_{\mathcal{B}_{ijh}}),\quad \forall  h\in\Omega(q+1).
\end{align}
Finally, recall that $x_i\notin\mathcal{H}(x_{\hat{\mathcal{S}}_{ij}})$ and $\widehat{\mathcal{S}}_{ij}={\mathcal{S}}_{ij}\setminus \{i\}$. 
Thus, $x_i\notin\mathcal{H}(\{y_h,~h=1,\cdots, (q+1)\})$, that is, $x_i\notin\mathcal{H}(\mathcal{Z}\setminus \{x_i\})$.
This further leads to $y^*\neq x_i$.} We have
\begin{align}\label{eq_ystartnnp}
	y^*\in \bigcap_{h\in\Omega(q+1)}\mathcal{H}(x_{\mathcal{B}_{ijh}})\setminus \{x_i\}\neq \emptyset.
\end{align}
Equation \eqref{eq_ystartnnp} clearly contradicts with the assumption that \eqref{HXomegabq+1} is false. We conclude that given \eqref{HXomegab}, equation \eqref{HXomegabq+1} must hold. Finally, by induction, when $q+1=b_i$, one has $\Omega(b_i) = \{1,\cdots,b_i\}$, that is,
\begin{align}
	\left(\bigcap_{k=1}^{b_i}\mathcal{H}(x_{\mathcal{B}_{ijk}})\right)\setminus \{x_i\}\neq \emptyset.
\end{align}
This completes the proof.
\end{proof}

\smallskip
\subsubsection{Proof of Lemma \ref{LM_ERW}a}
From the definition of convex hull, we know each $\mathcal{H}(x_{\mathcal{B}_{ijk}})$ is a convex polytope in $\mathbb{R}^n$. Since $x_i\notin \mathcal{H}(x_{\widehat{\mathcal{S}}_{ij}})$ and ${\mathcal{B}_{ijk}}\subset {\mathcal{S}_{ij}}=\widehat{\mathcal{S}}_{ij}\bigcup \{i\}$, for all $\mathcal{H}(x_{\mathcal{B}_{ijk}})$, $k\in\{1,\cdots,b_i\}$, $x_i$ is a vertex of $\mathcal{H}(x_{\mathcal{B}_{ijk}})$. Now, let 
\begin{align}\label{eq_defQ}
	\mathcal{Q}_{ij}=\bigcap_{k=1}^{b_i}\mathcal{H}(x_{\mathcal{B}_{ijk}}).
\end{align}
Obviously, $\mathcal{Q}_{ij}$ is a convex polytope, $x_i\in\mathcal{Q}_{ij}$, and $x_i$ is one of the vertices of $\mathcal{Q}_{ij}$. 
Furthermore, from Lemma \ref{L_Bij}, we know $\mathcal{Q}_{ij}\setminus \{x_i\}\neq \emptyset$ and thus, $\mathcal{Q}_{ij}$ has at least two vertices. 
Since in \eqref{eq_defQ}, $\mathcal{Q}_{ij}$ is obtained by intersecting $\mathcal{H}(x_{\mathcal{B}_{ijk}})$, using supporting hyperplane theorem~\cite{HR:13AG}, it can be derived that all vertices of $\mathcal{Q}_{ij}$ must lie on some of the faces of $\mathcal{H}(x_{\mathcal{B}_{ijk}})$, for $k\in\{1,\cdots,b_i\}$. Followed by this, let ${v}_Q$ denote a vertex of $\mathcal{Q}_{ij}$, and let $\mathcal{Y}_1,\mathcal{Y}_2,\cdots$ denote all the faces of polytope $\mathcal{H}(x_{\mathcal{B}_{ijk}})$, $k\in\{1,\cdots,b_i\}$ containing ${v}_Q$ such that
\begin{align}\label{eq_defvq}
	{v}_Q\in\bigcap_{\ell=1,2,\cdots}\mathcal{Y}_\ell.
\end{align}
From \eqref{eq_defvq}, if ${v}_Q\neq x_i$, there must exist an $\ell^*$ such that $x_i\notin\mathcal{Y}_{\ell^*}$.
This along with the fact that $x_i\notin \mathcal{H}(x_{\widehat{\mathcal{S}}_{ij}})$ means the $\mathcal{Y}_{\ell^*}$ must be spanned by vectors in $x_{\widehat{\mathcal{S}}_{ij}}$, where $\widehat{\mathcal{S}}_{ij}=\mathcal{S}_{ij}\setminus \{i\}$ (examples for such $\mathcal{Y}_{\ell^*}$ are the faces of the green-yellow intersected areas in Fig. \ref{Fig_Alg_1}). 
Consequently,  ${v}_Q\in\mathcal{H}(x_{\widehat{\mathcal{S}}_{ij}})$, and
${v}_Q\in\mathcal{Q}_{ij}\bigcap
\mathcal{H}(x_{\widehat{\mathcal{S}}_{ij}})\neq\emptyset$. This completes the proof. \hfill\qed

\medskip
\subsubsection{Proof of Lemma \ref{LM_ERW}b}
First, consider the case of $x_i\in\mathcal{H}(x_{\widehat{\mathcal{S}}_{ij}})$, i.e., the step 8 of Algorithm \ref{Algorithm_Main}. Since $\phi_{ij}=x_{i}$ and the agent $i$ itself is a normal agent, thus, $\phi_{ij}$ is a resilient convex combination and \eqref{RCCphi} holds. 

Second, for the case of $x_i\notin\mathcal{H}(x_{\widehat{\mathcal{S}}_{ij}})$, i.e., the step 11 of Algorithm \ref{Algorithm_Main}, one has 
$$\phi_{ij}\in\left(\bigcap_{k=1}^{b_i}\mathcal{H}(x_{\mathcal{B}_{ijk}})\right)\bigcap
\mathcal{H}(x_{\widehat{\mathcal{S}}_{ij}})\subset\bigcap_{k=1}^{b_i}\mathcal{H}(x_{\mathcal{B}_{ijk}}).$$
Recall that $\mathcal{B}_{ijk}\subset \mathcal{S}_{ij}$, $i \in \mathcal{B}_{ijk}$, and $|\mathcal{B}_{ijk}|=\sigma_i-{\kappa}_i$. Since the number of Byzantine neighbors of agent $i$ is bounded by ${\kappa}_i$, among all $k=1,\cdots,b_i$, there must exist one $\mathcal{B}_{ijk^{\star}}$  which consists only normal neighbors. Then given $\phi_{ij}\in\mathcal{H}(x_{\mathcal{B}_{ijk^{\star}}})$, one has $\phi_{ij}$ is a resilient convex combination and \eqref{RCCphi} holds.
This completes the proof. \hfill\qed

\medskip
\subsubsection{Proof of Lemma \ref{LM_ERW}c}
From steps 8 and 11 of Algorithm \ref{Algorithm_Main}, obviously, ${\phi}_{ij}\in\mathcal{H}(x_{\widehat{\mathcal{S}}_{ij}})$. Then based on Carath\'eodory's  theorem\cite{WJ86Caratheodory}, we know that there exists a subset $\Psi\subset \widehat{\mathcal{S}}_{ij}$ with $|\Psi|\le (n+1)$, such that ${\phi}_{ij}\in\mathcal{H}(x_{\Psi})$. That is,
\begin{align}\label{EQ_lm1}
	{\phi}_{ij}=\sum_{\ell\in\Psi}{\bar\beta}_{ij\ell}x_{\ell},\quad \sum_{\ell\in\Psi}{\bar\beta}_{ij\ell}=1,\quad {\bar\beta}_{ij\ell}\ge 0.
\end{align}
Then, since  $$\sum_{\ell\in\Psi}{\bar\beta}_{ij\ell}=1, \quad{\bar\beta}_{ij\ell}\ge 0 \quad \text{and} \quad |\Psi|\le (n+1),$$ 
there exists at least one $\ell^*\in\Psi$, such that ${\bar\beta}_{ij\ell^*}\ge\frac{1}{n+1}$.
Bringing this to equation \eqref{eq_weight}, for $l\in\widehat{\mathcal{S}}_{ij}$, if $\ell\in {\Psi}\subset \widehat{\mathcal{S}}_{ij} $, let ${\beta}_{ij\ell}={\bar\beta}_{ij\ell}$; if  $\ell\notin {\Psi}$
let ${\beta}_{ij\ell}=0$, otherwise. Then, it follows that  ${\beta}_{ij\ell^*}={\bar\beta}_{ij\ell^*}\ge\frac{1}{n+1}$. This completes the proof. \hfill\qed

\bibliographystyle{IEEEtran}

\bibliography{RES.bib}

\begin{IEEEbiography}
	[{\includegraphics[width=1in,height=1.25in,clip,keepaspectratio]{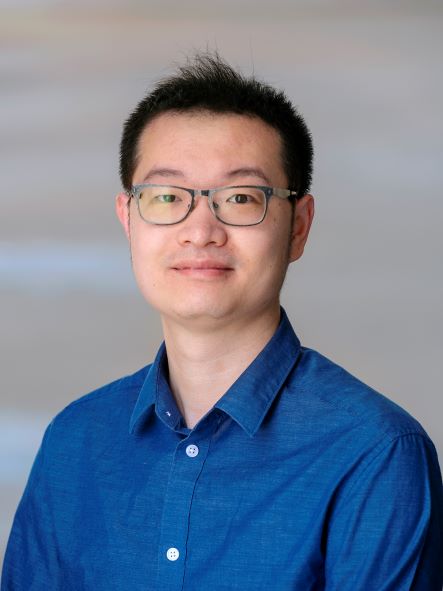}}] {Xuan Wang} 
	is an assistant professor with the Department of Electrical and Computer Engineering at George Mason University. 
	He received his Ph.D. degree in autonomy and control, from the School of Aeronautics and Astronautics, Purdue University in 2020. He was a post-doctoral researcher with the Department of Mechanical and Aerospace Engineering at the University of California, San Diego from 2020 to 2021. His research interests include multi-agent control and optimization; resilient multi-agent coordination; system identification and data-driven control of network dynamical systems.
\end{IEEEbiography}

\begin{IEEEbiography}
	[{\includegraphics[width=1in,height=1.25in,clip,keepaspectratio]{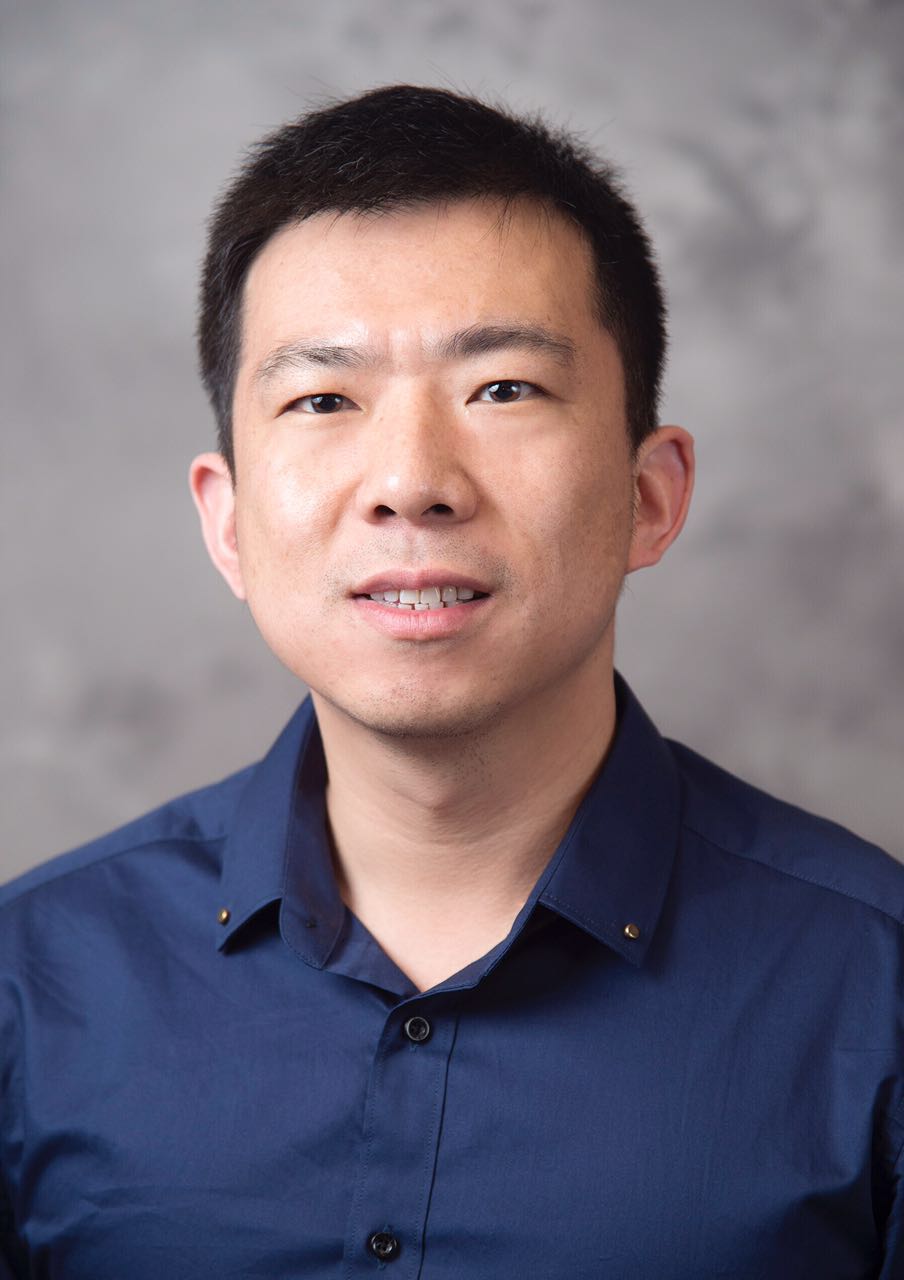}}] {Shaoshuai Mou} received a Ph. D. degree from Yale University in 2014. After working as a postdoctoral associate at MIT, he joined Purdue University as an assistant professor in the School of Aeronautics and Astronautics in 2015. His research interests include distributed algorithms for control/optimizations/learning, multi-agent networks, UAV collaborations, perception and autonomy, resilience and cyber-security. 
\end{IEEEbiography}

\begin{IEEEbiography}
	[{\includegraphics[width=1in,height=1.25in,clip,keepaspectratio]{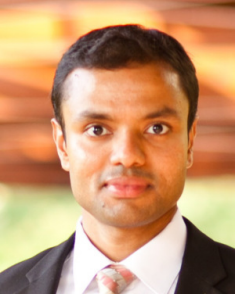}}] {Shreyas Sundaram}  is the Marie Gordon Professor in the Elmore Family School of Electrical and Computer Engineering at
	Purdue University. He received his PhD in Electrical Engineering from the University of Illinois at Urbana-Champaign in 2009. He was a Postdoctoral Researcher at the University of Pennsylvania from 2009 to 2010, and an Assistant Professor in the Department of Electrical and Computer Engineering at the University of Waterloo from 2010 to 2014. He is a recipient of the NSF CAREER award. At Purdue, he received the HKN Outstanding Professor Award, the Outstanding Mentor of Engineering Graduate Students Award, the Hesselberth Award for Teaching Excellence, and the Ruth and Joel Spira Outstanding Teacher Award. His research interests include network science, analysis of large-scale dynamical systems, fault-tolerant and secure control, linear system and estimation theory, and game theory. 
\end{IEEEbiography}

\end{document}